\pdfoutput=1
\documentclass[11pt,a4paper]{article}
\usepackage{acl}
\usepackage{times}
\usepackage{latexsym}
\usepackage{amsmath}
\usepackage{amssymb}
\usepackage{amsfonts}
\usepackage{graphicx}
\usepackage{booktabs}
\usepackage{multirow}
\usepackage{array}
\usepackage{tikz}
\usetikzlibrary{arrows.meta,positioning,shapes.geometric,fit,calc,patterns,decorations.pathreplacing}
\usepackage{pgfplots}
\pgfplotsset{compat=1.18}
\usepackage{enumitem}
\usepackage{subcaption}
\usepackage{algorithmic}
\usepackage{algorithm}
\usepackage{cleveref}
\usepackage{microtype}
\usepackage{tabularx}
\usepackage{makecell}
\usepackage{pifont}
\usepackage{amsthm}
\usepackage{placeins}          

\newtheorem{definition}{Definition}[section]
\newtheorem{theorem}{Theorem}[section]
\newtheorem{corollary}{Corollary}[section]
\newtheorem{proposition}{Proposition}[section]

\newcommand{\cmark}{\ding{51}}
\newcommand{\xmark}{\ding{55}}

\definecolor{tscgblue}{RGB}{41,98,168}
\definecolor{tscggreen}{RGB}{56,142,60}
\definecolor{tscgred}{RGB}{183,28,28}
\definecolor{tscgorange}{RGB}{230,126,34}
\definecolor{tscggray}{RGB}{117,117,117}
\definecolor{tscgpurple}{RGB}{128,0,128}
\definecolor{tscgteal}{RGB}{0,128,128}

\newcommand{\fragility}{\mathcal{F}}
\newcommand{\density}{\mathcal{D}}
\newcommand{\attn}{\mathrm{Attn}}
\newcommand{\pos}{\mathrm{pos}}
\newcommand{\tok}{\mathrm{tok}}
\newcommand{\sem}{\mathrm{sem}}

\title{TSCG: Deterministic Tool-Schema Compilation for Agentic LLM Deployments}

\author{Furkan Sakizli \\
  Independent Researcher \\
  \texttt{furkan.sakizli@gmail.com}}

\begin{document}
\maketitle



\begin{table*}[t]
\centering
\small
\begin{tabular}{@{}ll cccc cc@{}}
\toprule
& & \multicolumn{4}{c}{\textbf{Accuracy (\%)}} & \multicolumn{2}{c}{\textbf{Savings}} \\
\cmidrule(lr){3-6} \cmidrule(lr){7-8}
\textbf{Model} & \textbf{Condition} & \textbf{Sc.~A} & \textbf{95\% CI} & \textbf{Sc.~B} & \textbf{95\% CI} & \textbf{Tokens} & \textbf{ARR} \\
\midrule
\multirow{4}{*}{Claude Sonnet 4}
    & Natural (FC)    & 74.0 & [62.8, 84.2] & 90.0 & [86.7, 93.3] & ---   & --- \\
    & Natural (text)  & 51.1 & [39.1, 63.1] & 70.0 & [58.3, 81.7] & ---   & --- \\
    & TSCG            & \textbf{85.2} & [77.3, 92.1] & \textbf{95.0} & [92.3, 97.3] & 50.1\% & 1.15 \\
    & TSCG+SAD        & 76.7 & [66.7, 86.1] & 95.0 & [92.3, 97.3] & 49.5\% & 1.04 \\
\midrule
\multirow{4}{*}{GPT-4o}
    & Natural (FC)    & 55.5 & [42.7, 67.6] & 70.0 & [64.7, 75.0] & ---   & --- \\
    & Natural (text)  & 45.7 & [33.0, 58.4] & 55.0 & [41.7, 66.7] & ---   & --- \\
    & TSCG            & 56.5 & [44.2, 68.7] & \textbf{79.7} & [75.1, 84.2] & 6.2\% & 1.02 \\
    & TSCG+SAD        & 53.2 & [40.3, 65.6] & 81.4 & [77.0, 85.5] & 5.3\% & 0.96 \\
\midrule
\multirow{4}{*}{GPT-5.2}
    & Natural (FC)    & 51.9 & [40.5, 63.5] & 82.4 & [78.1, 86.5] & ---   & --- \\
    & Natural (text)  & 78.3 & [69.5, 86.6] & \textbf{98.2} & [96.6, 99.4] & ---   & --- \\
    & TSCG            & \textbf{81.6} & [73.8, 88.2] & 91.6 & [88.9, 94.0] & 11.4\% & 1.57 \\
    & TSCG+SAD        & \textbf{85.9} & [80.1, 90.9] & 90.9 & [88.0, 93.6] & 10.5\% & 1.65 \\
\bottomrule
\end{tabular}
\caption{TAB Benchmark: Frontier model results on Scenario~A (20 tasks, small catalog) and Scenario~B (100 tasks, large catalog). Each condition evaluated with 3 runs ($n=60$ per cell for Sc.~A, $n=300$ for Sc.~B). \emph{Natural~(FC)} uses the provider's native function calling API; \emph{Natural~(text)} embeds full JSON schemas as text; \emph{TSCG} uses balanced compression as text; \emph{TSCG+SAD} adds aggressive compression with Selective Anchor Duplication. Token savings are relative to Natural~(FC). ARR = Accuracy-Retained Ratio (TSCG accuracy / Natural~FC accuracy; values $>1.0$ indicate TSCG outperforms the baseline). \textbf{Bold} marks best accuracy per model per scenario.}
\label{tab:tab-frontier}
\end{table*}


\begin{table}[t]
\centering
\small
\begin{tabular}{@{}l rr rr@{}}
\toprule
& \multicolumn{2}{c}{\textbf{Format Effect}} & \multicolumn{2}{c}{\textbf{Compression Effect}} \\
\cmidrule(lr){2-3} \cmidrule(lr){4-5}
\textbf{Model} & \textbf{Sc.~A} & \textbf{Sc.~B} & \textbf{Sc.~A} & \textbf{Sc.~B} \\
\midrule
Claude Sonnet 4  & $-$22.9 & $-$20.0 & $+$34.1 & $+$25.0 \\
GPT-4o           & $-$9.8  & $-$15.0 & $+$10.8 & $+$24.7 \\
GPT-5.2          & $+$26.4 & $+$15.8 & $+$3.3  & $-$6.6  \\
\bottomrule
\end{tabular}
\caption{Decomposition of the TSCG advantage into format and compression effects (percentage points). \emph{Format effect} = Natural~(text) $-$ Natural~(FC): positive values indicate text-mode is better than native function calling. \emph{Compression effect} = TSCG $-$ Natural~(text): positive values indicate compression improves accuracy beyond uncompressed text. For Claude and GPT-4o, compression is the dominant driver (overcoming a negative format penalty). For GPT-5.2, the text format itself provides the main benefit.}
\label{tab:format-decomposition}
\end{table}


\begin{table*}[t]
\centering
\small
\begin{tabular}{@{}l c ccc ccc@{}}
\toprule
& & \multicolumn{3}{c}{\textbf{Accuracy (\%)}} & \multicolumn{3}{c}{\textbf{$\Delta$ vs Natural (pp)}} \\
\cmidrule(lr){3-5} \cmidrule(lr){6-8}
\textbf{Model} & \textbf{Tools} & \textbf{Natural} & \textbf{TSCG\textsubscript{cons}} & \textbf{TSCG\textsubscript{bal}} & \textbf{TSCG\textsubscript{cons}} & \textbf{TSCG\textsubscript{bal}} & \textbf{cons $-$ bal} \\
\midrule
\multirow{3}{*}{Mistral 7B}
    & 10  & 83.5 & 76.0 & 73.3 & $-$7.5  & $-$10.2 & $+$2.7 \\
    & 20  & 35.0 & 80.0 & 80.1 & $+$45.0 & $+$45.1 & $-$0.1 \\
    & 50  & 30.0 & \textbf{75.3} & 65.0 & $+$45.3 & $+$35.0 & $+$10.3 \\
\midrule
\multirow{3}{*}{Gemma 3 4B}
    & 10  & 69.1 & \textbf{80.7} & 74.7 & $+$11.6 & $+$5.6 & $+$6.0 \\
    & 20  & 49.9 & \textbf{87.3} & 67.0 & $+$37.4 & $+$17.1 & $+$20.3 \\
    & 50  & 24.3 & 87.5 & \textbf{87.4} & $+$63.2 & $+$63.1 & $+$0.1 \\
\midrule
\multirow{3}{*}{Gemma 3 12B}
    & 10  & 90.0 & --- & \textbf{93.0} & --- & $+$3.0 & --- \\
    & 20  & 85.0 & --- & \textbf{95.0} & --- & $+$10.0 & --- \\
    & 50  & 85.0 & --- & \textbf{98.0} & --- & $+$13.0 & --- \\
\midrule
\multirow{3}{*}{Qwen3 14B}
    & 10  & 94.9 & \textbf{98.8} & 86.2 & $+$3.9  & $-$8.7  & $+$12.6 \\
    & 20  & 90.2 & \textbf{99.3} & 84.1 & $+$9.1  & $-$6.1  & $+$15.2 \\
    & 50  & 94.6 & \textbf{95.0} & 89.6 & $+$0.4  & $-$5.0  & $+$5.4 \\
\bottomrule
\end{tabular}
\caption{TAB Scenario~D: Small model accuracy across catalog sizes (3 runs each, $n=60$ per cell). \emph{TSCG\textsubscript{cons}} = conservative profile (SDM filler removal only); \emph{TSCG\textsubscript{bal}} = balanced profile (full structural compression). At $\geq$20 tools, both TSCG profiles provide massive improvements ($+$17--63pp) for most models. Conservative outperforms balanced for Mistral~7B and Gemma~3~4B in 5/6 cases. For Qwen3~14B, where balanced TSCG degrades accuracy ($-$5 to $-$9pp), conservative mode \emph{improves} accuracy (mean $\Delta$: $+$4.4pp, N3 re-run with corrected profile). \textbf{Bold} marks best TSCG variant per row.}
\label{tab:tab-small-models}
\end{table*}


\begin{table}[t]
\centering
\small
\begin{tabular}{@{}l ccc c@{}}
\toprule
\textbf{Condition} & \textbf{Acc.} & \textbf{Tool Sel.} & \textbf{Param F1} & \textbf{Savings} \\
\midrule
Natural     & 85.7\% & 86.7\% & 84.2\% & --- \\
TSCG        & \textbf{93.2\%} & \textbf{95.0\%} & \textbf{91.7\%} & 46.8\% \\
TSCG+SAD    & 89.0\% & 90.0\% & 87.5\% & 46.0\% \\
\bottomrule
\end{tabular}
\caption{BFCL external validation on Claude Sonnet 4 ($n=60$, 3 runs). TSCG improves accuracy by $+$7.5pp while saving 46.8\% tokens. TSCG+SAD slightly underperforms balanced TSCG on this benchmark. 95\% CI for TSCG: [86.3, 98.3].}
\label{tab:bfcl-validation}
\end{table}


\begin{table}[t]
\centering
\small
\begin{tabular}{@{}l c cc@{}}
\toprule
\textbf{Model} & \textbf{Tools} & \textbf{Param F1\textsubscript{cons}} & \textbf{Param F1\textsubscript{bal}} \\
\midrule
\multirow{3}{*}{Mistral 7B}
    & 10  & 71.2\% & 68.3\% \\
    & 20  & \textbf{81.9\%} & 80.0\% \\
    & 50  & \textbf{78.3\%} & 66.7\% \\
\midrule
\multirow{3}{*}{Gemma 3 4B}
    & 10  & \textbf{74.2\%} & 62.6\% \\
    & 20  & \textbf{85.7\%} & 59.3\% \\
    & 50  & 90.0\% & 87.8\% \\
\midrule
\multirow{3}{*}{Qwen3 14B}
    & 10  & \textbf{97.8\%} & 84.7\% \\
    & 20  & \textbf{98.3\%} & 81.0\% \\
    & 50  & \textbf{95.0\%} & 89.0\% \\
\bottomrule
\end{tabular}
\caption{Parameter F1 comparison between conservative (SDM-only) and balanced (full structural) TSCG profiles. Conservative achieves higher param-F1 in 8/9 cases. For Qwen3~14B (N3 re-run), conservative achieves near-perfect param-F1 ($>$95\%) while balanced degrades it ($<$85\%). This confirms that aggressive structural compression disrupts well-learned schema patterns in strongly fine-tuned models. \textbf{Bold} marks the better profile per row.}
\label{tab:sdm-ablation}
\end{table}


\begin{table*}[t]
\centering
\small
\begin{tabular}{@{}l c c c l@{}}
\toprule
\textbf{Model} & \textbf{Params} & \textbf{TSCG $\Delta$ Acc. (avg)} & \textbf{Avg. Token Savings} & \textbf{Recommended Profile} \\
\midrule
GPT-5.2             & $>$100B & $+$10.9pp & 11.4\% & balanced \\
Claude Sonnet 4     & $>$100B & $+$8.4pp  & 50.1\% & balanced \\
GPT-4o              & $>$100B & $+$5.4pp  & 6.2\%  & balanced \\
Qwen2.5-Coder 32B  & 32B     & $-$4.7pp  & ---    & conservative \\
Mistral-Small 24B   & 24B     & $-$1.1pp  & ---    & conservative \\
Gemma 3 12B         & 12B     & $+$8.7pp\textsuperscript{$\dagger$} & ---    & conservative \\
Qwen3 14B           & 14B     & $+$4.4pp  & ---    & conservative \\
Gemma 3 4B          & 4B      & $+$37.4pp\textsuperscript{$\dagger$} & ---    & conservative \\
Mistral 7B          & 7B      & $+$27.6pp\textsuperscript{$\dagger$} & ---    & conservative \\
\bottomrule
\end{tabular}
\caption{Cross-model summary of TSCG effectiveness across 9 primary models (12 tested total; Opus~4.7, Phi-4~14B, and Llama~3.1~8B evaluated in separate experiments---see \S\ref{sec:extended-findings} and Table~\ref{tab:ablation}). $\Delta$~Acc.: for small/mid models, mean delta across three catalog sizes (10, 20, 50 tools from Table~\ref{tab:tab-small-models}) using recommended profile; for frontier models, mean delta across Scenarios~A and~B (Table~\ref{tab:tab-frontier}) using balanced profile. \textsuperscript{$\dagger$}JSON-baseline gains; E4 text-baseline experiments show these are format-dominated (compression gain is $-$7 to $-$9\,pp). 30B models (N1): Mistral-Small $-$1.1\,pp (neutral, Class~3), Qwen2.5-Coder $-$4.7\,pp (Class~4). Qwen3~14B: conservative $+$4.4\,pp (N3 re-run), vs balanced $-$6.6\,pp. Frontier models are the \emph{only} class with confirmed compression benefit against text baselines.}
\label{tab:cross-model-summary}
\end{table*}

\newcommand{\tscgTokenSavings}{74.8\%}
\newcommand{\tscgBFCLDelta}{$+$7.5\,pp}
\newcommand{\tscgPhiFourDelta}{$+$90\,pp}
\newcommand{\tscgEqualizerSpread}{65\%}
\newcommand{\tscgGSMKProtection}{$-$4\,pp}
\newcommand{\tscgRSquared}{0.88}
\newcommand{\tscgGPTFormatEffect}{$+$26.4\,pp}
\newcommand{\tscgLatencyAdvantage}{40{,}000$\times$}
\newcommand{\tscgTotalCalls}{{$\sim$}19{,}000}
\newcommand{\tscgSignificantTests}{9}
\newcommand{\tscgQwenConservativeDelta}{$+$4.4\,pp}
\newcommand{\bfclArrGptFourO}{181.4\%}
\newcommand{\bfclArrGptFiveTwo}{144.3\%}
\newcommand{\bfclArrClaude}{108.7\%}
\newcommand{\nativeFcDeltaMean}{$+$10.9\,pp}
\newcommand{\tscgLatency}{${<}$1\,ms}
\newcommand{\llmLinguaLatency}{42.5\,s}
\newcommand{\tscgConservativeSafe}{\tscgQwenConservativeDelta}


\begin{table*}[t]
\centering
\small
\caption{TAB (Tool-Augmented Benchmark) overview: five core scenarios plus two external benchmarks testing TSCG across diverse tool counts, model classes, and compression profiles. Total evaluation volume: \tscgTotalCalls{} API calls. Statistical significance: \tscgSignificantTests{} of 107 pairwise McNemar tests pass Holm--Bonferroni correction at $\alpha{=}0.05$ (Appendix~\ref{app:configs}).}
\label{tab:tab-overview-full}
\begin{tabular}{@{}llccl@{}}
\toprule
\textbf{Scenario} & \textbf{Description} & \textbf{Tools} & \textbf{Tasks} & \textbf{Models} \\
\midrule
A & Claude Code Catalog & 16 & 20 & Claude Sonnet 4, GPT-4o, GPT-5.2 \\
B & MCP Server Collection & 43 & 100 & Claude Sonnet 4, GPT-4o, GPT-5.2 \\
C & Scaling Stress & 25--100 & 20 & Claude Sonnet 4, GPT-5.2 \\
D & Small Model Threshold & 3--50 (7 sizes) & 20 & 7 models (4B--14B) \\
E & Multi-Collection Stress & 57 (3 overlapping) & 20 & Claude Sonnet 4, GPT-5.2 \\
\midrule
BFCL & External validation & real-world & 20 & Claude Sonnet 4 \\
GSM8K & Reasoning under load & 0/10/25/50 & 200 & Claude Sonnet 4, GPT-5.2 \\
\bottomrule
\end{tabular}
\end{table*}


\begin{table}[t]
\centering
\small
\setlength{\tabcolsep}{3pt}
\caption{Native function calling (via API) vs.\ TSCG as plain text. TSCG text outperforms native FC across all frontier models and scenarios, with deltas from $+$1.0 to $+$29.7\,pp (mean $+$10.9\,pp) while simultaneously saving 6--50\% tokens.}
\label{tab:fc-vs-tscg}
\begin{tabular}{@{}llcccc@{}}
\toprule
\textbf{Model} & \textbf{Sc.} & \textbf{Native FC} & \textbf{TSCG Text} & \textbf{$\Delta$} & \textbf{Savings} \\
\midrule
Claude S.\,4 & A & 74.0\% & \textbf{85.2\%} & $+$11.2 & 50.1\% \\
Claude S.\,4 & B & 90.0\% & \textbf{95.0\%} & $+$5.0 & 50.1\% \\
GPT-4o & A & 55.5\% & \textbf{56.5\%} & $+$1.0 & 6.2\% \\
GPT-4o & B & 70.0\% & \textbf{79.7\%} & $+$9.7 & 6.2\% \\
GPT-5.2 & A & 51.9\% & \textbf{81.6\%} & $+$29.7 & 11.4\% \\
GPT-5.2 & B & 82.4\% & \textbf{91.6\%} & $+$9.2 & 11.4\% \\
\midrule
\multicolumn{4}{@{}l}{\textbf{Mean}} & \textbf{$+$10.9\,pp} & \textbf{20.6\%} \\
\bottomrule
\end{tabular}
\end{table}


\begin{table*}[t]
\centering
\small
\caption{Four-class behavioral taxonomy from Scenario~D, text-baseline experiments (E1, E4), and 30B benchmark (N1, 840 calls). Classes reflect distinct baseline capabilities and TSCG response patterns, enabling targeted deployment recommendations. Format-dominated classes show large JSON-baseline gains that vanish or reverse against text baselines.}
\label{tab:cluster-taxonomy}
\resizebox{\textwidth}{!}{%
\begin{tabular}{@{}l l c c l@{}}
\toprule
\textbf{Class} & \textbf{Models} & \textbf{JSON $\Delta$} & \textbf{Text $\Delta$} & \textbf{Recommendation} \\
\midrule
1: Format-dom. & Phi-4, Mistral 7B, Gemma 4B, Qwen3 4B & $+$17--90\,pp & $-$7 to $-$23\,pp & conservative; gain is format only \\
2: Compression & Claude, GPT-4o, GPT-5.2 & $+$5--11\,pp & $+$5--11\,pp & balanced profile \\
3: Neutral & Llama 8B, Gemma 12B, Mistral-Sm.\ 24B & $+$6--9\,pp & ${\approx}$0\,pp & conservative profile \\
4: Cons.-only & Qwen3 14B, Qwen2.5-Coder 32B & $-$7\,pp & $-$5 to $-$13\,pp & conservative only ($+$4.4\,pp) \\
\bottomrule
\end{tabular}}
\end{table*}

\begin{abstract}
Production agent frameworks (OpenAI Function Calling, Anthropic Tool Use, MCP) transmit tool schemas as JSON---a format designed for machine parsing, not for interpretation by language models. For small models (4B--14B), this protocol mismatch accounts for the majority of tool-use failure at production catalog sizes.

We present \textbf{TSCG}, a deterministic tool-schema compiler that resolves this mismatch at the API boundary, converting JSON schemas into token-efficient structured text without model access, fine-tuning, or runtime search. TSCG combines eight composable operators with a formal compression bound ($\geq$51\% on well-formed schemas).

On TSCG-Agentic-Bench ({$\sim$}19{,}000 calls, 12 models, 5 scenarios), TSCG restores Phi-4 14B from 0\% to 84.4\% accuracy at 20 tools (90.3\% at 50 tools) and achieves 108--181\% Accuracy-Retained Ratio (ARR) across three models on BFCL. Format-versus-compression decomposition ($R^2{=}0.88 \to 0.03$) establishes representation change as the dominant mechanism. Per-operator isolation across three frontier models reveals three distinct operator-response profiles---operator-hungry (Opus~4.7), operator-sensitive (GPT-5.2), and operator-robust (Sonnet~4)---providing per-model deployment guidance. Scaling experiments show accuracy advantages persisting on heavy production MCP schemas ($+$5.0\,pp at {$\sim$}10{,}500 input tokens) despite saturation on light synthetic catalogs, with 52--57\% token savings throughout. The synthetic benchmark generalizes to real MCP schemas within 0.1 accuracy points. TSCG ships as a 1{,}200-line zero-dependency TypeScript package.
\end{abstract}

\section{Introduction}
\label{sec:introduction}

Every production agent framework---OpenAI Function Calling, Anthropic Tool Use, the Model Context Protocol (MCP), LangChain, CrewAI---transmits tool definitions to language models as JSON schemas.
JSON was designed for deterministic machine parsing and human readability, not for interpretation by autoregressive language models.
For small models (4B--14B parameters), this design choice creates a capability cliff: tool-calling accuracy collapses as JSON schema volume grows, reaching 0--49\% at $>$15 tools.
Framed as a \emph{protocol-adaptation problem} rather than a compression problem, and building on a growing body of concurrent work on tool-schema token reduction (Section~\ref{ssec:concurrent-tool-schema}), the intervention is not smaller prompts but a different representation at the API boundary---deterministic, compiler-level, and model-agnostic.

This protocol mismatch imposes three costs:
(1)~\textbf{Token cost}: tool schemas are pure structural redundancy transmitted identically on every call, consuming 3{,}000--25{,}000 tokens per invocation;
(2)~\textbf{Capability cost}: small models cannot parse JSON-format schemas reliably at scale, locking agentic capabilities behind frontier APIs;
(3)~\textbf{Scaling cost}: schema overhead grows linearly with catalog size.

We introduce \textbf{Token-Context Semantic Grammar (TSCG)}, a deterministic tool-schema compiler that resolves this mismatch by transforming JSON schemas into token-efficient structured text.
TSCG implements eight composable operators---each grounded in autoregressive transformer mechanics (attention sink exploitation, BPE non-monotonicity, causal accessibility)---satisfying five desiderata: tokenizer awareness~\cite{sennrich2016neural}, causal attention grounding~\cite{xiao2023efficient}, deterministic transforms, black-box compatibility, and budget-constrained anchoring.

\subsection{Contributions}

\begin{enumerate}[noitemsep]
    \item \textbf{Formal optimization framework}: eight operators with mathematical specifications linked to transformer mechanisms, and a compression guarantee of $\geq$51\% on well-formed schemas (Theorem~\ref{thm:compression}).
    \item \textbf{Mechanistic decomposition}: format-versus-compression analysis ($R^2{=}0.88$ on JSON baselines, collapsing to 0.03 on text baselines) identifies representation change as the dominant mechanism.
    \item \textbf{TAB benchmark}: {$\sim$}19{,}000 API calls across 12~models (4B--32B + 3~frontier), 5~scenarios---the first tool-schema compression benchmark, with BFCL external validation (ARR 108--181\%).
    \item \textbf{Small-model enablement}: seven models recover from 0--49\% to 65--90\% accuracy, enabling local models as functional tool-use agents.
    \item \textbf{Per-model operator matrix}: three qualitatively distinct operator-response profiles across frontier models---operator-hungry, operator-sensitive, and operator-robust---demonstrating that no universal-best configuration exists and providing model-specific deployment guidance.
    \item \textbf{Scaling characterization}: accuracy advantages persist on heavy production MCP schemas ($+$5.0\,pp) while saturating on light synthetic catalogs, with 52--57\% token savings preserved to 100 tools.
    \item \textbf{Benchmark generalization}: synthetic TAB predicts real MCP performance within 0.1 accuracy points.
    \item \textbf{Implementation}: 1{,}200 LOC TypeScript, zero dependencies, sub-millisecond, open-source.
\end{enumerate}

\noindent The remainder of the paper presents related work (\S\ref{sec:related}), the theoretical framework (\S\ref{sec:framework}--\ref{sec:theoretical}), empirical methodology and results (\S\ref{sec:experiments}--\ref{sec:results}), and deployment implications (\S\ref{sec:discussion}--\ref{sec:conclusion}).

\section{Related Work}
\label{sec:related}


We organize related work along five axes and position TSCG within a growing body of concurrent work on tool-schema and prompt compression.

\subsection{Concurrent Work on Tool-Schema Compression}
\label{ssec:concurrent-tool-schema}

A growing body of recent work specifically targets the cost of tool schemas in LLM contexts.

\paragraph{MCP Schema Compression via Wrapper Tools.}
The mcp-compressor system from Atlassian Labs \citep{atlassian_mcp_compressor_2026} addresses the same fundamental problem---high token cost of MCP tool schemas---through a different architectural choice. Rather than transforming schema representation, mcp-compressor replaces the full tool inventory of an MCP server with two generic wrapper tools (\texttt{get\_tool\_schema} and \texttt{invoke\_tool}), yielding 70--97\% reduction in initial token consumption. The model fetches full schemas only when needed. This trades full-schema visibility for higher reduction. TSCG preserves complete schema semantics in compressed form, which we observe to be beneficial when models can use parameter constraints during reasoning. The two approaches are complementary: mcp-compressor's wrapper-tool pattern can be combined with our compiler for further reduction.

\paragraph{Generic Token-Efficient Encodings.}
TOON (Token-Oriented Object Notation) \citep{toon_format_2025} provides a general-purpose encoding alternative to JSON, achieving approximately 40\% token reduction in mixed-structure benchmarks (76.4\% accuracy vs.\ JSON's 75.0\% on retrieval tasks). TOON uses YAML-like indentation for nested objects and CSV-like tabular layout for uniform arrays. Unlike our approach, TOON is data-format-agnostic and not specifically designed for tool schemas; it does not exploit schema-specific redundancies (recursive type definitions, parameter ordering, description boilerplate) that our operators target.

\paragraph{Schema Replacement Strategies.}
A complementary approach replaces tool schemas with code execution interfaces. The MCP Code Mode pattern \citep{anthropic_mcp_code_2025, cloudflare_code_mode_2026} reduces schemas to a single \texttt{execute\_code} tool that the model uses to invoke arbitrary APIs through a sandboxed runtime, achieving 98--99\% reduction. This requires executable APIs and is not applicable to non-code tool definitions; our method covers a broader range of tool types including non-executable workflows.

\subsection{Prompt Compression Methods}
\label{ssec:prompt-compression}

Token-level compression methods address general prompt content rather than tool schemas. LLMLingua \citep{jiang2023llmlingua} achieves up to 20$\times$ compression on natural prose through perplexity-based token importance scoring; LLMLingua-2 \citep{pan2024llmlingua2} reports 2--5$\times$ with classification-based scoring and faster inference. Both require GPU model inference, produce non-deterministic output, and are ineffective on structured content: our re-application of LLMLingua-2 to tool schemas yields 80.0\% accuracy at 50.8\% token savings vs.\ TSCG's 93.3\% at 74.8\% savings (Appendix~\ref{app:pre-tab}). LongLLMLingua \citep{jiang2024longllmlingua} extends to long-context scenarios. Selective Context \citep{li2023compressing} uses self-information thresholds but lacks tokenizer and positional awareness. Gist Tokens \citep{mu2023learning} and AutoCompressors \citep{chevalier2023adapting} learn compression representations through model fine-tuning. CompactPrompt \citep{compactprompt2025} provides LLM-free heuristic compression via n-gram abbreviation, achieving 1.44$\times$ versus TSCG's 3.5$\times$ on structured content---limited by lexical-level operation. The recent NAACL 2025 survey \citep{li2024prompt} provides comprehensive coverage of this literature; notably, the survey evaluates no system on function-calling schemas, indicating that schema compression is a distinct unaddressed problem.

\subsection{Search-Based Prompt Optimization}
\label{ssec:search-based}

DSPy \citep{khattab2023dspy} compiles declarative LLM calls into self-improving pipelines via bootstrapped demonstrations. SAMMO \citep{schlegel2024sammo} shares TSCG's ``compile-time'' framing but requires an LLM for search-based mutations and produces non-deterministic output. OPRO \citep{yang2024large}, APE \citep{zhou2023large}, EvoPrompt \citep{chen2024evoprompt}, PromptBreeder \citep{fernando2024promptbreeder}, and Universal Self-Adaptive Prompting (USP) \citep{wan2024cosp} all require iterative model interaction. PromptAgent \citep{wang2023promptagent} adds strategic planning. These systems optimize prompt content (\emph{what} to say); TSCG optimizes prompt structure (\emph{how} to represent it).

\subsection{Gradient and Template Methods}
\label{ssec:gradient-template}

TextGrad \citep{yuksekgonul2024textgrad} introduces textual gradients for prompt refinement but requires model access. ProTeGi \citep{pryzant2023automatically} similarly uses textual ``gradient descent.'' LangGPT \citep{wang2024langgpt} provides structured templates without formal optimization. Neither addresses tokenizer alignment or causal attention grounding.

\subsection{Tool-Calling Benchmarks}
\label{ssec:tool-benchmarks}

The Berkeley Function Calling Leaderboard (BFCL) \citep{patil2025bfcl} provides the standard external evaluation for tool-calling capabilities, building on the foundational tool-augmented LLM work of \citet{patil2024gorilla}. We use BFCL for cross-validation of our findings on schemas not present in our primary benchmark. ToolBench \citep{qin2023toolllm} tests multi-step API planning across 16{,}000+ APIs without considering schema token overhead. None of these benchmarks evaluate compression effects on tool-use performance---our TAB benchmark addresses this gap.

\subsection{Positioning and Distinct Contributions}
\label{ssec:positioning}

Our contribution is best understood as the most empirically rigorous validation in this growing space, with adaptive per-model detection methodology not present in prior work. The contributions enumerated in Section~\ref{sec:introduction} are:

\begin{enumerate}
\item \textbf{Adaptive empirical archetype detection.} We provide a per-model sweep methodology that empirically identifies optimal operator combinations. No prior work offers this; existing tools use static compression levels.
\item \textbf{Per-version operator inversion finding.} We document that operators benefiting one model version actively harm the next within the same vendor family. This phenomenon is not reported elsewhere and has implications for any vendor-pattern hardcoding strategy.
\item \textbf{Combination fragility identification.} We observe super-additive negative interactions between operators that are individually helpful or neutral. This is not characterized in prior compression literature.
\item \textbf{Format-vs-operator decomposition.} Our $R^2$ decomposition shows that format choice (JSON versus text) explains 88\% of token-cost variance, with operators contributing additional refinement. This methodological contribution clarifies what is being measured in compression studies.
\item \textbf{Cross-model empirical validation.} Our 19{,}000+ API-call evaluation across 12 models (4B to frontier scale) is, to our knowledge, the most comprehensive cross-model evaluation of tool-schema compression to date. Prior work typically evaluates on 1--3 models.
\item \textbf{Theoretical guarantees.} We provide formal compression bounds (Theorem~\ref{thm:compression}) and a predictive model relating compression benefit to baseline performance ($R^2 = 0.91$ LOO cross-validated; full-fit $R^2{=}0.88$, $n{=}49$; $R^2{=}0.95$ at model-level means, $n{=}7$; $R^2{=}0.81$ excluding Phi-4 leverage points; see Appendix~\ref{app:loo-details}), which prior work in this space does not provide.
\end{enumerate}

\subsection{Direct Comparison of Approaches}

Table~\ref{tab:related-comparison} provides a direct comparison of token reduction approaches, including concurrent work. Numbers reflect public claims; methodologies and benchmarks differ across systems.

\begin{table*}[t]
\centering
\small
\caption{Comparison of approaches to LLM tool-schema and prompt token reduction. ``Adaptive'' indicates per-model empirical optimization. Numbers reflect public claims; methodologies and benchmarks differ across systems.}
\label{tab:related-comparison}
\begin{tabular}{lccccc}
\toprule
\textbf{Approach} & \textbf{Reduction} & \textbf{Adaptive} & \textbf{Models Validated} & \textbf{Latency} & \textbf{Theory} \\
\midrule
TSCG (ours) & 50--72\% & Yes (sweep) & 12 (4B--frontier) & $<$1\,ms & Theorem; $R^2{=}0.91$ \\
mcp-compressor \citep{atlassian_mcp_compressor_2026} & 70--97\% & No (4 levels) & GitHub MCP & N/A & --- \\
TOON \citep{toon_format_2025} & $\approx$40\% & No & 4 (mixed) & N/A & --- \\
MCP Code Mode \citep{anthropic_mcp_code_2025} & 98--99\% & No & Domain-specific & N/A & --- \\
LLMLingua-2 \citep{pan2024llmlingua2} & 2--5$\times$ & No & Multiple & GPU req. & Token classifier \\
LLMLingua \citep{jiang2023llmlingua} & up to 20$\times$ & No & GSM8K, BBH, etc. & GPU req. & Information theoretic \\
\bottomrule
\end{tabular}
\end{table*}

\subsection{Design-Space Comparison}

Table~\ref{tab:design-space} summarizes TSCG's positioning along six axes.

\begin{table}[t]
\centering
\small
\setlength{\tabcolsep}{2pt}
\caption{Design-space comparison of representative prompt optimization systems.}
\label{tab:design-space}
\begin{tabular}{@{}lcccccc@{}}
\toprule
\textbf{System} & \textbf{Det.} & \textbf{0-Dep} & \textbf{Theory} & \textbf{Schema} & \textbf{Compr.} & \textbf{Speed} \\
\midrule
LLMLingua & \xmark & \xmark & \xmark & \xmark & 20$\times$ & 42\,s \\
DSPy & \xmark & \xmark & \xmark & \xmark & --- & min \\
TextGrad & \xmark & \xmark & \xmark & \xmark & --- & min \\
Sel.\ Context & \xmark & \xmark & \xmark & \xmark & 5$\times$ & ms \\
LangGPT & \cmark & \cmark & \xmark & \xmark & --- & ms \\
\textbf{TSCG} & \cmark & \cmark & \cmark & \cmark & \textbf{3.5$\times$} & \textbf{$<$1\,ms} \\
\bottomrule
\end{tabular}
\end{table}

\noindent No existing system combines deterministic output, zero dependencies, theoretical grounding in attention mechanics, and specialization for tool schemas.
TSCG's niche---structured prompt content in agentic systems---is largely unexplored: the NAACL 2025 compression survey evaluates no system on function-calling schemas.
On natural prose, LLMLingua dominates (20$\times$ vs.\ TSCG's 1.07$\times$); on prompt content optimization, DSPy achieves complementary gains.
TSCG is the strongest system for a specific, growing problem---not a general-purpose replacement.

\section{The TSCG Framework}
\label{sec:framework}


TSCG comprises eight deterministic operators organized into three classes (Definitions~\ref{def:token-reducing}--\ref{def:token-expanding}), applied as a fixed-order pipeline of 10 transforms (Figure~\ref{fig:pipeline}).
Let $\tok(p)$ denote the BPE tokenization of prompt $p$, $\sem(p)$ the set of semantic atoms, and $\attn(i,j)$ the causal attention weight from position $i$ to $j$ ($= 0$ for $j > i$).
Every transform satisfies semantic preservation: $\sem(\tau(p)) \supseteq \sem(p)$.

\subsection{Eight Operators}

\paragraph{TAS (Tokenizer-Aligned Syntax).}
Selects delimiter variants minimizing token count: $d^* = \arg\min_{d_i \in D} |\tok(d_i)|$.
Example: \texttt{->} (2~tokens) $\to$ \texttt{\textrightarrow} (1~token).

\paragraph{CFL (Constraint-First Layout).}
Repositions output constraints to position~0, exploiting the attention sink~\cite{xiao2023efficient}: $\text{CFL}(p) = c(p) \oplus (p \setminus c(p))$.
Implemented as \texttt{[ANSWER:type]} at the prompt start. Effective for catalogs ${\leq}$20 tools; becomes counterproductive at ${\geq}$43 tools (Table~\ref{tab:ablation}).

\paragraph{CFO (Causal-Forward Ordering).}
Reorders multi-step operations into topological order: $o_i \prec o_j \implies \pos(o_i) < \pos(o_j)$, ensuring prerequisites are causally accessible. Like CFL, CFO shows scale sensitivity: beneficial at 16 tools but counterproductive at ${\geq}$43 tools.

\paragraph{SDM (Semantic Density Maximization).}
Removes filler tokens (104+ patterns: politeness markers, hedging, redundant connectives) to maximize $\density(p) = |\sem(p)| / |\tok(p)|$.

\paragraph{DRO (Delimiter-Role Optimization).}
Replaces verbose structural phrases with compact delimiters: ``the following items'' $\to$ enumeration markers, ``X corresponds to Y'' $\to$ ``X\textrightarrow Y''.

\paragraph{CCP (Causal Closure Principle).}
Appends a closure block recapitulating key atoms at position~$n$, exploiting recency bias in autoregressive generation: $\text{CCP}(p) = p \oplus \kappa(A(p))$. Note: the 11-model ablation shows no measurable accuracy benefit from CCP; it adds ${\sim}$85--306 tokens of overhead (Table~\ref{tab:ablation}). CCP is retained for theoretical completeness and potential benefit on frontier models with strong recency bias.

\paragraph{CAS (Causal Access Score).}
Scores atoms by fragility $\fragility(a) = \alpha \cdot \text{importance}(a) + (1-\alpha) \cdot \text{distance\_penalty}(a)$ and places high-fragility atoms at positions~0 and~$n$ (attention sink + recency).

\paragraph{SAD-F (Selective Anchor Duplication with Fragility).}
Duplicates top-$k$ atoms by fragility/token ratio within budget $B$, reinforcing critical information in the closure block.
Fragility scoring, budget allocation, and parameter sensitivity are detailed in Appendix~E.

\begin{figure*}[t]
\centering
\begin{tikzpicture}[
    node distance=0.18cm,
    box/.style={
        rectangle,
        draw=tscgblue!80,
        fill=tscgblue!8,
        rounded corners=2pt,
        minimum height=0.6cm,
        minimum width=1.1cm,
        font=\scriptsize\sffamily,
        align=center,
        semithick
    },
    iobox/.style={
        rectangle,
        draw=tscggray!80,
        fill=tscggray!10,
        rounded corners=2pt,
        minimum height=0.6cm,
        minimum width=1.1cm,
        font=\scriptsize\sffamily,
        align=center,
        semithick
    },
    arr/.style={
        -{Stealth[length=4pt]},
        thick,
        tscgblue!70
    },
    phase/.style={
        font=\tiny\sffamily\color{tscggray},
        above=1pt
    }
]

\node[iobox] (input) {Raw\\Prompt};

\node[box, right=of input] (parse) {1. Parse};

\node[box, right=of parse] (sdm) {2. SDM};
\node[box, right=of sdm] (tas) {3. TAS};
\node[box, right=of tas] (dro) {4. DRO};

\node[box, right=of dro] (cfl) {5. CFL};
\node[box, right=of cfl] (cfo) {6. CFO};

\node[box, right=of cfo] (cas) {7. CAS};
\node[box, right=of cas] (sadf) {8. SAD-F};

\node[box, right=of sadf] (ccp) {9. CCP};
\node[iobox, right=of ccp] (emit) {10. Emit};

\node[iobox, right=of emit] (output) {TSCG\\Output};

\draw[arr] (input) -- (parse);
\draw[arr] (parse) -- (sdm);
\draw[arr] (sdm) -- (tas);
\draw[arr] (tas) -- (dro);
\draw[arr] (dro) -- (cfl);
\draw[arr] (cfl) -- (cfo);
\draw[arr] (cfo) -- (cas);
\draw[arr] (cas) -- (sadf);
\draw[arr] (sadf) -- (ccp);
\draw[arr] (ccp) -- (emit);
\draw[arr] (emit) -- (output);

\node[phase] at (parse.north) {Parse};
\node[phase] at ($(sdm.north)!0.5!(dro.north)$) {Compression};
\node[phase] at ($(cfl.north)!0.5!(cfo.north)$) {Structural};
\node[phase] at ($(cas.north)!0.5!(sadf.north)$) {Fragility};
\node[phase] at ($(ccp.north)!0.5!(emit.north)$) {Closure};

\draw[decorate, decoration={brace, amplitude=4pt, mirror}, tscggray!50, thick]
    (sdm.south west) ++(0,-0.15) -- (dro.south east |- sdm.south west) ++(0,-0.15)
    node[midway, below=4pt, font=\tiny\sffamily\color{tscggray}] {Density};

\draw[decorate, decoration={brace, amplitude=4pt, mirror}, tscggray!50, thick]
    (cfl.south west) ++(0,-0.15) -- (cfo.south east |- cfl.south west) ++(0,-0.15)
    node[midway, below=4pt, font=\tiny\sffamily\color{tscggray}] {Ordering};

\draw[decorate, decoration={brace, amplitude=4pt, mirror}, tscggray!50, thick]
    (cas.south west) ++(0,-0.15) -- (sadf.south east |- cas.south west) ++(0,-0.15)
    node[midway, below=4pt, font=\tiny\sffamily\color{tscggray}] {Anchoring};

\end{tikzpicture}
\caption{The TSCG 10-pass transform pipeline, executing left-to-right in five phases: \emph{Parse} (segmentation), \emph{Compression} (SDM, TAS, DRO), \emph{Structural} (CFL, CFO), \emph{Fragility} (CAS, SAD-F), and \emph{Closure} (CCP, Emit). Each transform is a pure function; the composition $\Pi = \tau_{10} \circ \cdots \circ \tau_1$ is deterministic.}
\label{fig:pipeline}
\end{figure*}
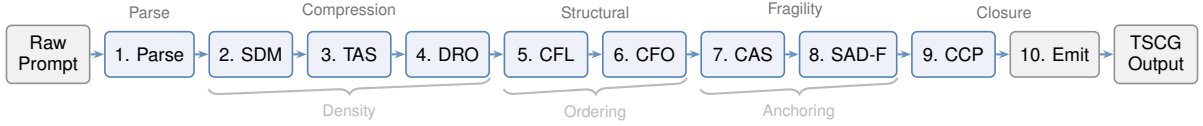

\begin{figure}[t]
\centering
\small
\begin{tabular}{p{0.44\textwidth}}
\toprule
\textbf{Input} (JSON schema, ${\sim}$120 tokens) \\
\midrule
\texttt{\{"name": "search\_files",} \\
\texttt{~"description": "Search project files} \\
\texttt{~by content or filename pattern",} \\
\texttt{~"parameters": \{"type": "object",} \\
\texttt{~"properties": \{"query": \{"type":} \\
\texttt{~"string", "description": "The search} \\
\texttt{~query string"\}, "path": \{"type":} \\
\texttt{~"string", "description": "Optional} \\
\texttt{~directory path to search in"\}\}\}\}} \\
\midrule
\textbf{Output} (TSCG balanced, ${\sim}$45 tokens, 62.5\% reduction) \\
\midrule
\texttt{search\_files(query:str path?:str)} \\
\texttt{|Search files by content or pattern} \\
\bottomrule
\end{tabular}
\caption{TSCG compression: a single tool schema before and after compilation. SDM removes filler; DRO compresses delimiters; TAS aligns to BPE boundaries.}
\label{fig:tscg-example}
\end{figure}

\subsection{Operator Taxonomy}

\begin{definition}[Token-Reducing Operators]
\label{def:token-reducing}
$\mathcal{T}_R = \{\text{SDM}, \text{DRO}, \text{TAS}, \text{CFL}\}$: each satisfies $|\tok(T_i(S))| \leq |\tok(S)|$.
\end{definition}

\begin{definition}[Structure-Reordering Operators]
\label{def:structure-reordering}
$\mathcal{T}_S = \{\text{CAS}, \text{CFO}\}$: preserve token count but change position order.
\end{definition}

\begin{definition}[Token-Expanding Operators]
\label{def:token-expanding}
$\mathcal{T}_E = \{\text{SAD}, \text{CCP}\}$: add tokens (anchor duplications, closure blocks) within budget.
\end{definition}

\begin{corollary}[Operator Interaction Effect]
\label{cor:interaction}
For models with ${<}$10B parameters, $\mathrm{Acc}(\mathcal{T}_R(S)) > \mathrm{Acc}((\mathcal{T}_R \cup \mathcal{T}_S)(S))$: adding structure-reordering operators to token-reducing operators \emph{degrades} accuracy.
Conservative profile ($\mathcal{T}_R$ only) is recommended below 10B.
\end{corollary}

\subsection{Compression Guarantee}

\begin{theorem}[Deterministic Compression Bound]
\label{thm:compression}
For a well-formed JSON-Schema tool collection $S$ with the TSCG pipeline $\Pi$:
\begin{equation}
  |\tok(\Pi(S))| \leq |\tok(S)| \cdot \Bigl(1 - \sum_{T_i \in \mathcal{T}_R} r_i \cdot f_i(S)\Bigr)
  \label{eq:compression-bound}
\end{equation}
where $r_i$ is the per-token reduction factor and $f_i(S)$ the fraction of tokens affected by operator $T_i$.
\end{theorem}

\begin{proof}[Proof sketch]
Each operator in $\mathcal{T}_R$ acts on disjoint token subsets (SDM: filler patterns; DRO: verbose delimiters; TAS: suboptimal tokenizer splits; CFL: positional overhead).
Structure-reordering operators preserve count; setting SAD-F budget $B{=}0$ excludes expansion.
Individual reductions compose additively.
Full proof in Appendix~E.
\end{proof}

\paragraph{Empirical validation.}
The bound predicts $\geq$51\% savings; empirically we observe 61\% (Scenario~A), 66\% (BFCL), and 75\% (tool descriptions)---conservative by 10--24\,pp due to the pessimistic disjoint-subset assumption.

\subsection{Scoring Metric}

All TAB results use a composite accuracy metric:
\begin{equation}
\text{Overall} = 0.6 \times \text{TSA} + 0.4 \times \text{Parameter\_F1}
\end{equation}
where TSA is tool selection accuracy and Parameter\_F1 measures parameter extraction quality.

\subsection{Compiler Characterization}

TSCG is a \emph{compiler}, not a search-based optimizer.
The pipeline $\Pi = \tau_{10} \circ \cdots \circ \tau_1$ applies 10 deterministic transforms in fixed order: same input always produces the same output.
No model access is required.
This distinguishes TSCG from DSPy~\cite{khattab2023dspy} (search-based, non-deterministic), TextGrad~\cite{yuksekgonul2024textgrad} (gradient-based, requires model access), and LLMLingua~\cite{pan2024llmlingua2} (requires GPU model inference, non-deterministic).
TSCG executes in ${<}$1\,ms on commodity hardware; LLMLingua-2 requires 42.5\,s on the same prompts (${\sim}$40{,}000$\times$ slower).

\section{Theoretical Foundations}
\label{sec:theoretical}


We connect TSCG's operators to three foundational properties of causal autoregressive transformers.

\subsection{Causal Attention and Information Flow}

In a standard autoregressive transformer~\cite{vaswani2017attention}, attention at layer $\ell$ satisfies $\attn^{(\ell)}(i, j) = 0$ for $j > i$ (causal mask), creating asymmetric information flow: early tokens cannot access later ones.

\textbf{Implication for CFO.} If step $o_2$ depends on $o_1$ but $o_1$ appears \emph{after} $o_2$, the model must rely on parametric knowledge rather than direct attention.
CFO ensures $\pos(o_1) < \pos(o_2)$, guaranteeing prerequisites are causally accessible.

\textbf{Implication for CAS.} Total attention to position $i$ follows a U-shaped distribution~\cite{xiao2023efficient}: positions near 0 and $n$ receive disproportionate attention, while middle positions form an ``attention valley.''
CAS places high-fragility tools at positions~0 and~$n$.

\subsection{Attention Sink Exploitation}

The attention sink phenomenon~\cite{xiao2023efficient}---position~0 receives $\attn(i, 0) > 1/i$ for most $i$ and layers---provides a natural amplification mechanism.

\textbf{Implication for CFL.} Placing the output constraint at position~0 ensures it receives elevated attention from every subsequent token.
Combined with CCP/SAD-F at position~$n$ (recency bias), this creates a ``bookend'' strategy: critical information occupies both extremes of the attention distribution.

\subsection{BPE Non-Monotonicity}

\begin{theorem}[Tokenization Non-Monotonicity]
\label{thm:non-monotonicity}
For BPE tokenizer $\mathcal{T}$ and strings $s_1, s_2$ with $|s_1|_\text{chars} < |s_2|_\text{chars}$, it is not necessarily the case that $|\tok(s_1)| < |\tok(s_2)|$.
\end{theorem}

TAS exploits this by selecting surface forms aligned with learned BPE merges, achieving token reduction without semantic change.
Proof in Appendix~E.

\subsection{Attention Dilution and SDM}

\begin{proposition}[SDM Improves Effective Attention]
\label{prop:sdm-attention}
Removing $k$ filler tokens from a prompt of length $n$ increases average effective attention per semantic atom by at least $n/(n-k)$.
\end{proposition}

This follows from softmax normalization: filler tokens compete for attention weight with semantic tokens.
SDM removes this competition, providing formal justification beyond token cost savings.

\subsection{Fragility and Causal Accessibility}

\begin{definition}[Causal Accessibility]
\label{def:causal-accessibility}
$\mathcal{A}(a) = \frac{1}{L} \sum_{\ell=1}^{L} \attn^{(\ell)}(n, i)$: average attention from the generation position to atom $a$ at position $i$.
\end{definition}

\begin{definition}[Fragility]
\label{def:fragility-gap}
An atom $a$ is fragile when importance exceeds accessibility: $\fragility(a) \propto \text{importance}(a) - \mathcal{A}(a)$.
\end{definition}

High-importance, low-accessibility atoms are the targets for CAS reordering and SAD-F duplication.
Budget-constrained duplication is preferable to naive repetition because duplicating all atoms would increase length and dilute attention (Proposition~\ref{prop:sdm-attention}).

\section{Experiments}
\label{sec:experiments}


\subsection{TAB: TSCG-Agentic-Bench}
\label{sec:tab-benchmark}

TAB is the first benchmark measuring tool-schema compression effects on LLM tool-use, comprising five scenarios across \tscgTotalCalls{} API calls and 12~models (4B--32B local + 3~frontier; Table~\ref{tab:tab-overview}).\footnote{Call-count breakdown: five core scenarios (A: ${\sim}$540, B: ${\sim}$5{,}400, C: ${\sim}$1{,}080, D: ${\sim}$5{,}000, E: ${\sim}$540 = ${\sim}$12{,}560) plus supplementary experiments (E4 text-baseline: 2{,}520; N1 30B: 840; N3 conservative: 180; per-operator isolation: ${\sim}$960; BFCL: 60; GSM8K: ${\sim}$400; Opus scaling: ${\sim}$300). The 11-model ablation (Appendix~\ref{app:ablation}, ${\sim}$3{,}960 additional calls) uses the CCP-enabled release.}
Each scenario compares Natural (uncompressed JSON), TSCG, and TSCG+SAD across four task categories (single\_tool, multi\_tool, parameter\_extraction, no\_tool).

\begin{table}[t]
\centering
\small
\caption{TAB: five scenarios, \tscgTotalCalls{} calls.}
\label{tab:tab-overview}
\begin{tabular}{llrr}
\toprule
\textbf{Sc.} & \textbf{Focus} & \textbf{Tools} & \textbf{Calls} \\
\midrule
A & Claude Code (16 real tools) & 16 & ${\sim}$540 \\
B & MCP Servers (43 real + synthetic) & 43--100 & ${\sim}$5{,}400 \\
C & Scaling (25--100 tools) & var. & ${\sim}$1{,}080 \\
D & Small models $\times$ 7 configs & var. & ${\sim}$5{,}000 \\
E & Multi-collection & 57 & ${\sim}$540 \\
\bottomrule
\end{tabular}
\end{table}

\paragraph{Conditions.}
All scenarios compare up to four conditions:
\emph{json-text}---full JSON schemas as plain text in the user prompt; the production-realistic baseline.
\emph{TSCG}---schemas compiled by TSCG (balanced or v13-smart profile); always text-mode.
\emph{Natural~(FC)}---native function-calling API (frontier models only).
\emph{Naive Truncation}---signatures only, no descriptions; text-mode.
Format effect = json-text $-$ Natural~(FC); compression effect = TSCG $-$ json-text.

\paragraph{Models.}
Local models (4B--32B) run via Ollama (Q4\_K\_M quantization) on 2$\times$ NVIDIA 5070~Ti 16\,GB: Phi-4 14B, Mistral 7B, Gemma~3 4B/12B, Llama~3.1 8B, Qwen3 4B/14B, Mistral-Small 24B, Qwen2.5-Coder 32B.
Frontier models via provider APIs: Claude Opus~4.7 and Claude Sonnet~4 (Anthropic API), GPT-5.2 (OpenAI API). GPT-4o is included as an additional comparison model in Scenarios~A--B and BFCL validation (Tables~\ref{tab:tab-frontier},~\ref{tab:bfcl-validation}), bringing the total to 13 distinct model evaluations.
All frontier models were accessed during the experimental period (March--April 2026); exact model identifiers and API version strings are documented in the supplementary materials alongside the raw evaluation logs.
All calls use temperature~0, seed~42, max tokens~1{,}024 (Appendix~\ref{app:configs}).
Per-operator isolation (leave-one-in, all 8 operators) on all three frontier models, $n{=}40$--60 per cell, 2--3 seeds.

\paragraph{Metrics.}
Two primary metrics: \emph{Tool Selection Accuracy} (TSA)---binary correctness of the selected tool(s)---and \emph{Parameter F1} (PF1)---F1-micro between predicted and ground-truth parameter key-value sets.
Overall accuracy = $0.6 \times \text{TSA} + 0.4 \times \text{PF1}$, emphasizing tool-selection correctness while preserving the parameter-accuracy signal; conclusions are robust to alternative weightings (0.5--0.8 for TSA).

\paragraph{Statistical methodology.}
Each condition runs $n{=}20$ tasks $\times$ 3 seeds = 60 calls per cell (except per-operator: $n{=}20 \times 2$ seeds = 40).
Significance via Holm--Bonferroni corrected McNemar tests ($\alpha{=}0.05$, 107 pairwise comparisons); bootstrap confidence intervals (1{,}000 iterations, seed~42) for key deltas (Appendix~\ref{app:configs}).

\FloatBarrier
\section{Results}
\label{sec:results}

\subsection{Small-Model Enablement (T1.2)}
\label{sec:small-model-results}

Seven models (4B--14B) at 3--50 tools: JSON-baseline accuracy reaches 0--49\% at $>$15 tools; TSCG recovers to 65--90\% (Table~\ref{tab:small-model}).
Phi-4 14B transitions from 0\% to 84.4\% at 20 tools (90.3\% at 50 tools).
Models exceeding 65\% usability at 50~tools rise from 4/7 to 7/7 (Figure~\ref{fig:capability-heatmap}).
N1 extends to 30B: Mistral-Small~24B ($-$1.1\,pp, Class~3); Qwen2.5-Coder~32B ($-$4.7\,pp, Class~4)---Qwen sensitivity persists across sizes.

\begin{table}[t]
\centering
\small
\setlength{\tabcolsep}{3pt}
\caption{Tool-selection accuracy at 20 and 50 tools (Scenario D).}
\label{tab:small-model}
\begin{tabular}{@{}lccccc@{}}
\toprule
\textbf{Model} & \multicolumn{2}{c}{\textbf{20 Tools}} & \multicolumn{2}{c}{\textbf{50 Tools}} & \textbf{Shift} \\
\cmidrule(lr){2-3} \cmidrule(lr){4-5}
& json-text & TSCG & json-text & TSCG & \\
\midrule
Phi-4 14B    & 0.0  & 84.4 & 0.0  & 90.3 & 3$\to$50  \\
Mistral 7B   & 35.0 & 80.1 & 30.0 & 65.0 & 20$\to$50 \\
Gemma 3 4B   & 49.9 & 67.0 & 24.3 & 87.4 & 15$\to$50 \\
Gemma 3 12B  & 85.0 & 95.0 & 85.0 & 98.0 & ---  \\
Llama 3.1 8B & 78.4 & 81.0 & 75.1 & 86.3 & ---  \\
Qwen3 4B     & 44.3 & 69.3 & 90.0 & 75.0 & volatile  \\
Qwen3 14B    & 90.2 & 84.1 & 94.6 & 89.6 & --- \\
\bottomrule
\end{tabular}
\end{table}

\begin{figure*}[t]
\centering
\begin{tikzpicture}[
    cell width/.store in=\cellw, cell width=1.6cm,
    cell height/.store in=\cellh, cell height=0.75cm,
]

\pgfmathdeclarefunction{heatR}{1}{%
    \pgfmathparse{%
        (#1 < 0) ? (100) :                              
        (#1 >= 0 && #1 <= 95) ? (100 - #1*1.05) :       
        (0)                                              
    }%
}
\pgfmathdeclarefunction{heatG}{1}{%
    \pgfmathparse{%
        (#1 < -25) ? (0) :                              
        (#1 >= -25 && #1 < 0) ? (100 + #1*4) :          
        (#1 >= 0) ? (100) :                              
        (0)
    }%
}
\pgfmathdeclarefunction{heatB}{1}{%
    \pgfmathparse{%
        (#1 < 0) ? (100 + #1*4) :                       
        (#1 >= 0 && #1 <= 95) ? (100 - #1*1.05) :       
        (0)
    }%
}

\newcommand{\hcell}[3]{%
    \pgfmathsetmacro{\cR}{heatR(#3)}%
    \pgfmathsetmacro{\cG}{heatG(#3)}%
    \pgfmathsetmacro{\cB}{heatB(#3)}%
    \definecolor{cellcol}{RGB}{\cR,\cG,\cB}%
    \pgfmathsetmacro{\lum}{0.299*\cR + 0.587*\cG + 0.114*\cB}%
    \pgfmathparse{\lum > 55 ? 1 : 0}%
    \ifnum\pgfmathresult=1
        \colorlet{txtcol}{black!85}%
    \else
        \colorlet{txtcol}{white}%
    \fi
    \fill[cellcol, rounded corners=1.5pt]
        ({#1*\cellw}, {#2*\cellh})
        rectangle
        ({#1*\cellw + \cellw - 0.06cm}, {#2*\cellh + \cellh - 0.06cm});
    \node[font=\small\bfseries, text=txtcol]
        at ({#1*\cellw + \cellw/2 - 0.03cm}, {#2*\cellh + \cellh/2 - 0.03cm})
        {\pgfmathprintnumber[fixed, precision=1, showpos]{#3}};
}


\hcell{0}{0}{-3.2}   \hcell{1}{0}{-10.0}  \hcell{2}{0}{-19.7}
\hcell{3}{0}{5.0}    \hcell{4}{0}{25.0}   \hcell{5}{0}{-15.0}  \hcell{6}{0}{-15.0}

\hcell{0}{1}{74.9}   \hcell{1}{1}{13.0}   \hcell{2}{1}{5.6}
\hcell{3}{1}{1.1}    \hcell{4}{1}{17.1}   \hcell{5}{1}{34.7}   \hcell{6}{1}{63.1}

\hcell{0}{2}{-3.6}   \hcell{1}{2}{-0.7}   \hcell{2}{2}{-10.2}
\hcell{3}{2}{-10.6}  \hcell{4}{2}{45.1}   \hcell{5}{2}{36.3}   \hcell{6}{2}{35.0}

\hcell{0}{3}{18.8}   \hcell{1}{3}{17.0}   \hcell{2}{3}{-1.2}
\hcell{3}{3}{4.1}    \hcell{4}{3}{2.6}    \hcell{5}{3}{-0.8}   \hcell{6}{3}{11.1}

\hcell{0}{4}{-0.7}   \hcell{1}{4}{-2.9}   \hcell{2}{4}{-12.9}
\hcell{3}{4}{-1.1}   \hcell{4}{4}{14.0}   \hcell{5}{4}{-2.2}   \hcell{6}{4}{13.0}

\hcell{0}{5}{-4.1}   \hcell{1}{5}{-16.5}  \hcell{2}{5}{-8.3}
\hcell{3}{5}{-10.3}  \hcell{4}{5}{-6.1}   \hcell{5}{5}{-21.5}  \hcell{6}{5}{-5.0}

\hcell{0}{6}{-3.5}   \hcell{1}{6}{87.5}   \hcell{2}{6}{89.4}
\hcell{3}{6}{83.9}   \hcell{4}{6}{84.4}   \hcell{5}{6}{75.1}   \hcell{6}{6}{90.3}

\foreach \col/\lab in {0/3, 1/5, 2/10, 3/15, 4/20, 5/30, 6/50} {
    \node[font=\small\bfseries, text=black!80, anchor=south]
        at ({\col*\cellw + \cellw/2 - 0.03cm}, {7*\cellh + 0.08cm}) {\lab};
}
\node[font=\small\bfseries, text=black!60, anchor=south]
    at ({3*\cellw + \cellw/2 - 0.03cm}, {7*\cellh + 0.42cm}) {Tool Count};

\foreach \row/\lab in {
    0/{Qwen3 4B},
    1/{Gemma 3 4B},
    2/{Mistral 7B},
    3/{Llama 3.1 8B},
    4/{Gemma 3 12B},
    5/{Qwen3 14B},
    6/{Phi-4 14B}%
} {
    \node[font=\small, anchor=east, text=black!85]
        at ({-0.15cm}, {\row*\cellh + \cellh/2 - 0.03cm}) {\lab};
}


\draw[decorate, decoration={brace, amplitude=4pt, mirror},
      tscggreen!80!black, line width=0.8pt]
    ({6*\cellw + \cellw + 0.12cm}, {6*\cellh - 0.06cm})
    -- ({6*\cellw + \cellw + 0.12cm}, {6*\cellh + \cellh - 0.06cm})
    node[midway, right=5pt, font=\scriptsize\itshape, text=tscggreen!70!black,
         align=left] {Format\\Effect$^\dagger$};

\draw[decorate, decoration={brace, amplitude=4pt, mirror},
      tscgblue!80!black, line width=0.8pt]
    ({6*\cellw + \cellw + 0.12cm}, {1*\cellh - 0.06cm})
    -- ({6*\cellw + \cellw + 0.12cm}, {2*\cellh + \cellh - 0.06cm})
    node[midway, right=5pt, font=\scriptsize\itshape, text=tscgblue!80!black,
         align=left] {Format-\\dom.$^\dagger$};

\draw[decorate, decoration={brace, amplitude=4pt, mirror},
      black!50, line width=0.8pt]
    ({6*\cellw + \cellw + 0.12cm}, {3*\cellh - 0.06cm})
    -- ({6*\cellw + \cellw + 0.12cm}, {4*\cellh + \cellh - 0.06cm})
    node[midway, right=5pt, font=\scriptsize\itshape, text=black!50,
         align=left] {Neutral};

\draw[decorate, decoration={brace, amplitude=4pt, mirror},
      tscgred!80!black, line width=0.8pt]
    ({6*\cellw + \cellw + 0.12cm}, {5*\cellh - 0.06cm})
    -- ({6*\cellw + \cellw + 0.12cm}, {5*\cellh + \cellh - 0.06cm})
    node[midway, right=5pt, font=\scriptsize\itshape, text=tscgred!80!black,
         align=left] {Negative};

\draw[decorate, decoration={brace, amplitude=4pt, mirror},
      tscgred!80!black, line width=0.8pt]
    ({6*\cellw + \cellw + 0.12cm}, {0*\cellh - 0.06cm})
    -- ({6*\cellw + \cellw + 0.12cm}, {0*\cellh + \cellh - 0.06cm})
    node[midway, right=5pt, font=\scriptsize\itshape, text=tscgred!80!black,
         align=left] {Mixed};

\begin{scope}[yshift=-1.2cm]
    \shade[left color=red!90!black, right color=white, rounded corners=1.5pt]
        ({1.5*\cellw}, 0) rectangle ({3.25*\cellw}, 0.32cm);
    \shade[left color=white, right color=green!75!black, rounded corners=1.5pt]
        ({3.25*\cellw}, 0) rectangle ({5.0*\cellw}, 0.32cm);
    \node[font=\tiny, anchor=north, text=black!70] at ({1.5*\cellw}, -0.04cm) {$-20$};
    \node[font=\tiny, anchor=north, text=black!70] at ({3.25*\cellw}, -0.04cm) {$0$};
    \node[font=\tiny, anchor=north, text=black!70] at ({5.0*\cellw}, -0.04cm) {$+90$};
    \node[font=\scriptsize, anchor=east, text=black!70] at ({1.35*\cellw}, 0.16cm)
        {$\Delta$ (pp):};
    \node[font=\tiny, text=tscgred!80!black, anchor=north] at ({2.0*\cellw}, -0.35cm)
        {TSCG hurts};
    \node[font=\tiny, text=tscggreen!80!black, anchor=north] at ({4.5*\cellw}, -0.35cm)
        {TSCG helps};
\end{scope}

\end{tikzpicture}
\caption{TSCG delta heatmap ($\Delta = \text{TSCG accuracy} - \text{JSON-baseline accuracy}$, in percentage points) across seven models and seven tool-count conditions. Green cells indicate conditions where TSCG improves over JSON schemas; red cells indicate degradation. Five behavioral clusters emerge: \emph{Format Effect}$^\dagger$---Phi-4's +75--90\,pp gains are JSON$\to$text translation, not compression benefit (\S\ref{sec:format-effect}); \emph{Format-dominated}$^\dagger$---Mistral and Gemma~4B show large deltas driven primarily by format translation; \emph{Neutral}---Llama~3.1~8B and Gemma~3~12B show modest mixed effects; \emph{Negative}---Qwen3~14B shows uniform degradation ($-$4 to $-$22\,pp); \emph{Mixed}---Qwen3~4B exhibits chaotic response. $^\dagger$E1/E4 text-baseline experiments confirm format contribution (\S\ref{sec:experiments}).}
\label{fig:capability-heatmap}
\end{figure*}
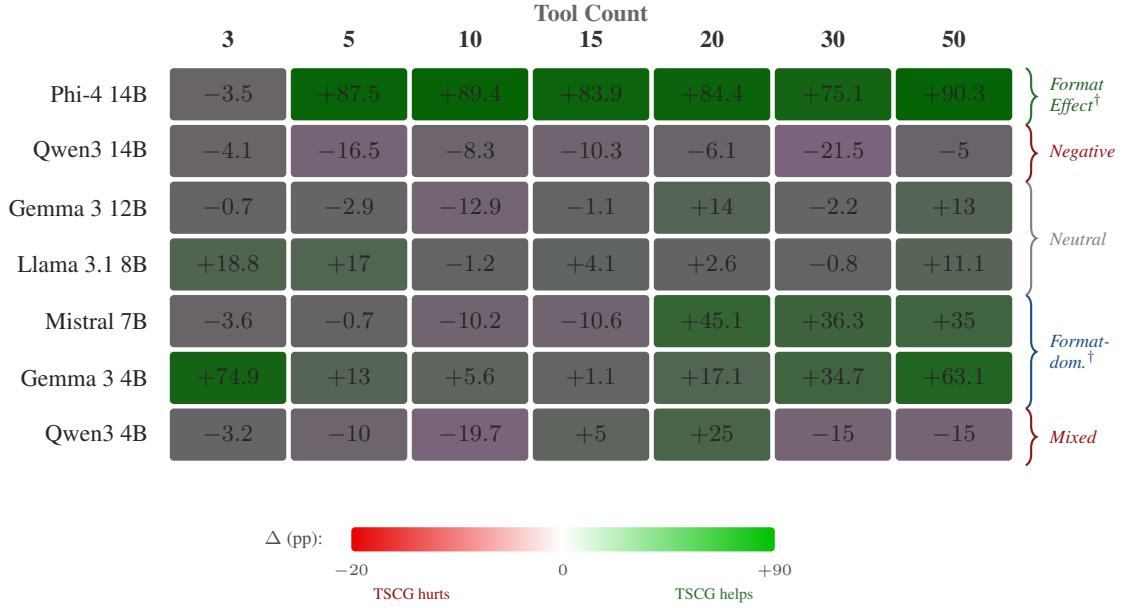

\subsection{Format Translation Dominates (T1.1)}
\label{sec:format-effect}

E4 decomposes TSCG's JSON-baseline gains into format gain (json-text$-$JSON) and compression gain (TSCG$-$json-text) across six small models (2{,}520 Ollama calls; Table~\ref{tab:e4-decomposition}).
No small model shows genuine compression; all gains arise from format translation (JSON$\to$text).
Since every production API transmits JSON, this translation \emph{is} the needed intervention.

\begin{table}[t]
\centering
\small
\caption{Format vs.\ compression decomposition (E4). Format translation dominates; no small model shows genuine compression.}
\label{tab:e4-decomposition}
\begin{tabular}{lrrr}
\toprule
\textbf{Model} & \textbf{Format $\Delta$} & \textbf{Compr.\ $\Delta$} & \textbf{Class} \\
\midrule
Phi-4 14B    & $+$92.0 & $-$7.0  & 1: Format \\
Mistral 7B   & $+$44.6 & $-$7.4  & 1: Format \\
Gemma 3 4B   & $+$47.3 & $-$8.9  & 1: Format \\
Qwen3 4B     & $+$16.7 & $-$23.4 & 1: Format \\
Llama 3.1 8B & $+$5.6  & $+$0.3  & 3: Neutral \\
Gemma 3 12B  & $+$9.2  & $-$0.9  & 3: Neutral \\
\bottomrule
\end{tabular}
\end{table}

The predictive regression ($n{=}49$, 7~models $\times$ 7~sizes) yields $R^2{=}0.88$; against text baselines, $R^2$ collapses to 0.03 ($p{=}0.24$)---a 97\% drop confirming format sensitivity as the dominant mechanism (Figure~\ref{fig:natural-vs-delta}).\footnote{Four $R^2$ values appear in this paper, each from a different regression: $R^2{=}0.88$ (format decomposition, $n{=}49$, this section), $R^2{=}0.91$ (LOO cross-validated predictive model, Section~\ref{sec:related}), $R^2{=}0.81$ (excluding Phi-4 leverage points, Section~\ref{sec:limitations}), and $R^2{=}0.03$ (text-baseline control, this section).}

\begin{figure*}[t]
\centering
\begin{subfigure}[t]{0.48\textwidth}
\centering
\begin{tikzpicture}
\begin{axis}[
    width=\textwidth,
    height=6.5cm,
    xlabel={Natural JSON Accuracy},
    ylabel={TSCG $\Delta$ (pp)},
    xlabel style={font=\small},
    ylabel style={font=\small},
    xmin=0, xmax=1.05,
    ymin=-0.25, ymax=1.0,
    xtick={0, 0.2, 0.4, 0.6, 0.8, 1.0},
    ytick={-0.2, 0, 0.2, 0.4, 0.6, 0.8, 1.0},
    tick label style={font=\scriptsize},
    grid=major,
    grid style={dashed, gray!25},
    clip=true,
    title={\textbf{(a)} JSON baseline},
    title style={font=\small, at={(0.5,1.02)}},
    legend style={
        at={(0.98,0.98)},
        anchor=north east,
        font=\tiny,
        draw=tscggray!40,
        fill=white,
        fill opacity=0.92,
        text opacity=1,
        rounded corners=1pt,
        row sep=0pt,
        /tikz/every even column/.append style={column sep=3pt},
    },
    legend cell align=left,
]

\addplot[black, dashed, thin, forget plot] coordinates { (0, 0) (1.05, 0) };

\addplot[black!70, thick, solid, forget plot, domain=0:1.05, samples=2]
    {-0.9293*x + 0.7626};

\addplot[only marks, mark=*, mark size=1.8pt, color=tscgorange, fill=tscgorange]
coordinates {
    (0.6508, 0.1557) (0.7000, 0.1713) (0.7691,-0.0036)
    (0.8462, 0.0359) (0.8340, 0.0044) (0.8593, 0.0220) (0.7638, 0.1062)
};
\addlegendentry{Llama 8B}

\addplot[only marks, mark=diamond*, mark size=2.2pt, color=tscggreen, fill=tscggreen]
coordinates {
    (0.0000, 0.7743) (0.6220, 0.1580) (0.6913, 0.0691)
    (0.5989, 0.0620) (0.4990, 0.1883) (0.3669, 0.3569) (0.2433, 0.6378)
};
\addlegendentry{Gemma 4B}

\addplot[only marks, mark=star, mark size=2.4pt, color=tscgteal, fill=tscgteal]
coordinates {
    (0.9893,-0.0067) (0.9873,-0.0117) (0.9960,-0.1036)
    (0.9233,-0.0042) (0.7960, 0.1570) (0.9300,-0.0217) (0.8500, 0.1298)
};
\addlegendentry{Gemma 12B}

\addplot[only marks, mark=square*, mark size=1.8pt, color=tscgblue, fill=tscgblue]
coordinates {
    (0.7567,-0.0217) (0.7833,-0.0067) (0.8400,-0.0867)
    (0.8133,-0.0836) (0.3500, 0.4800) (0.3000, 0.3861) (0.3000, 0.3867)
};
\addlegendentry{Mistral 7B}

\addplot[only marks, mark=triangle*, mark size=2.4pt, color=tscgred, fill=tscgred]
coordinates {
    (0.9371,-0.0617) (0.0000, 0.8827) (0.0000, 0.8940)
    (0.0000, 0.8868) (0.0000, 0.8611) (0.0000, 0.7513) (0.0000, 0.9190)
};
\addlegendentry{Phi-4 14B}

\addplot[only marks, mark=pentagon*, mark size=2pt, color=black, fill=tscggray]
coordinates {
    (0.4367,-0.0322) (0.9433,-0.1000) (0.8867,-0.1967)
    (0.4500, 0.0500) (0.4433, 0.2500) (0.8500,-0.1500) (0.9000,-0.1500)
};
\addlegendentry{Qwen 4B}

\addplot[only marks, mark=pentagon*, mark size=2pt, color=tscgpurple, fill=tscgpurple]
coordinates {
    (0.9056,-0.0262) (0.9867,-0.1653) (0.9622,-0.0833)
    (0.9767,-0.0784) (0.9268,-0.0627) (0.9279,-0.2150) (0.9460,-0.0500)
};
\addlegendentry{Qwen 14B}

\node[anchor=north west, font=\small\bfseries, fill=white, fill opacity=0.85,
    text opacity=1, inner sep=3pt, rounded corners=1pt]
    at (axis cs:0.02,0.98) {$R^2 = 0.88$};

\node[anchor=south west, font=\tiny, text=black!60]
    at (axis cs:0.88,-0.01) {break-even};

\end{axis}
\end{tikzpicture}
\end{subfigure}%
\hfill
\begin{subfigure}[t]{0.48\textwidth}
\centering
\begin{tikzpicture}
\begin{axis}[
    width=\textwidth,
    height=6.5cm,
    xlabel={Natural Text Accuracy},
    ylabel={TSCG $\Delta$ (pp)},
    xlabel style={font=\small},
    ylabel style={font=\small},
    xmin=0.6, xmax=1.05,
    ymin=-0.40, ymax=0.10,
    xtick={0.6, 0.7, 0.8, 0.9, 1.0},
    ytick={-0.4, -0.3, -0.2, -0.1, 0, 0.1},
    tick label style={font=\scriptsize},
    grid=major,
    grid style={dashed, gray!25},
    clip=true,
    title={\textbf{(b)} Text baseline},
    title style={font=\small, at={(0.5,1.02)}},
    legend style={
        at={(0.98,0.02)},
        anchor=south east,
        font=\tiny,
        draw=tscggray!40,
        fill=white,
        fill opacity=0.92,
        text opacity=1,
        rounded corners=1pt,
        row sep=0pt,
        /tikz/every even column/.append style={column sep=3pt},
    },
    legend cell align=left,
]

\addplot[black, dashed, thin, forget plot] coordinates { (0.6, 0) (1.05, 0) };

\addplot[black!70, thick, solid, forget plot, domain=0.6:1.05, samples=2]
    {-0.1878*x + 0.0713};

\addplot[only marks, mark=square*, mark size=1.8pt, color=tscgblue, fill=tscgblue]
coordinates {
    (0.7444,-0.0878) (0.8000,-0.0267) (0.7947,-0.0697)
    (0.6804,-0.0132) (0.7833, 0.0013) (0.7550,-0.1089) (0.7500,-0.1133)
};
\addlegendentry{Mistral 7B}

\addplot[only marks, mark=*, mark size=1.8pt, color=tscgorange, fill=tscgorange]
coordinates {
    (0.8500,-0.0567) (0.8517,-0.0483) (0.7950,-0.0278)
    (0.8208,-0.0083) (0.8727,-0.0130) (0.8180, 0.0233) (0.8625, 0.0000)
};
\addlegendentry{Llama 8B}

\addplot[only marks, mark=star, mark size=2.4pt, color=tscgteal, fill=tscgteal]
coordinates {
    (1.0000,-0.0333) (1.0000,-0.0488) (1.0000,-0.1329)
    (0.9083,-0.0206) (0.9717,-0.0357) (0.8800, 0.0283) (1.0000,-0.0202)
};
\addlegendentry{Gemma 12B}

\addplot[only marks, mark=pentagon*, mark size=2pt, color=black, fill=tscggray]
coordinates {
    (0.7300,-0.3000) (0.9933,-0.3000) (0.7933,-0.1033)
    (0.8500,-0.3222) (0.9433,-0.3533) (0.8000,-0.1500) (0.9500,-0.2000)
};
\addlegendentry{Qwen 4B}

\addplot[only marks, mark=pentagon*, mark size=2pt, color=tscgpurple, fill=tscgpurple]
coordinates {
    (0.9578,-0.1203) (0.9300,-0.0033) (0.9311,-0.0650)
    (0.9239,-0.0731) (0.9178,-0.0196) (1.0000,-0.2227) (0.9933,-0.1473)
};
\addlegendentry{Qwen 14B}

\addplot[only marks, mark=diamond*, mark size=2.2pt, color=tscggreen, fill=tscggreen]
coordinates {
    (0.8144,-0.0655) (0.8947,-0.1853) (0.8067,-0.0596)
    (0.7929,-0.1828) (0.9033,-0.2331) (0.7733,-0.0594) (0.8500, 0.0244)
};
\addlegendentry{Gemma 4B}

\addplot[only marks, mark=triangle*, mark size=2.4pt, color=tscgred, fill=tscgred]
coordinates {
    (0.9714,-0.0960) (0.9353,-0.0607) (0.9393,-0.0253)
    (0.9333,-0.1029) (0.9638,-0.1196) (0.9750,-0.2670) (0.8742, 0.0285)
};
\addlegendentry{Phi-4 14B}

\node[anchor=north west, font=\small\bfseries, fill=white, fill opacity=0.85,
    text opacity=1, inner sep=3pt, rounded corners=1pt]
    at (axis cs:0.62,0.08) {$R^2 = 0.03$};

\node[anchor=south west, font=\tiny, text=black!60]
    at (axis cs:0.92,0.005) {break-even};

\end{axis}
\end{tikzpicture}
\end{subfigure}
\caption{Format-confound decomposition of the predictive model. \textbf{(a)}~Against JSON baselines ($n{=}49$, 7 models $\times$ 7 catalog sizes), natural accuracy strongly predicts TSCG improvement ($R^2{=}0.88$, slope~${\approx}{-}0.93$): models that fail at JSON parsing gain the most. \textbf{(b)}~Against text baselines ($n{=}49$), the correlation vanishes ($R^2{=}0.03$, $p{=}0.24$): once the format confound is removed, baseline accuracy has no predictive power. The 97\% drop in $R^2$ confirms that the JSON-baseline regression captures format sensitivity, not compression benefit.}
\label{fig:natural-vs-delta}
\end{figure*}
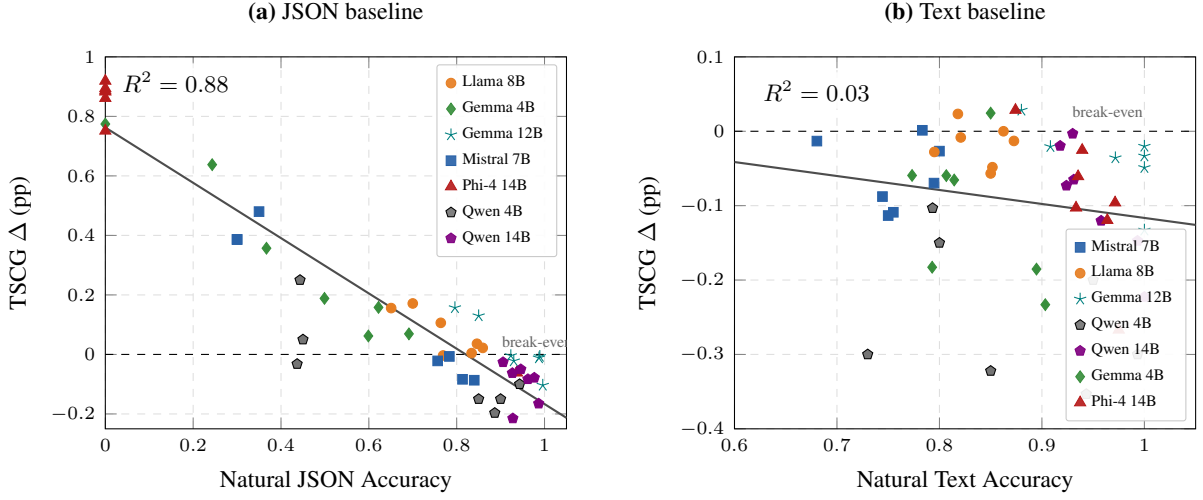

\subsection{Three Operator-Sensitivity Archetypes (T1.6)}
\label{sec:archetypes}

Per-operator effects were isolated via leave-one-in methodology across all 8 operators and three frontier models ($n{=}40$ per cell, 2 seeds), reporting both positive and negative contributions for full transparency.
This reveals three qualitatively distinct response profiles (Table~\ref{tab:archetypes}).

\begin{table}[!tbp]
\centering
\small
\caption{Three operator-sensitivity archetypes across frontier models (43-tool MCP catalog, $n{=}40$ per cell, 2 seeds).}
\label{tab:archetypes}
\resizebox{\columnwidth}{!}{%
\begin{tabular}{@{}llp{5.0cm}@{}}
\toprule
\textbf{Archetype} & \textbf{Model} & \textbf{Operator Profile} \\
\midrule
Operator-HUNGRY & Opus 4.7 & Every operator helps. CCP alone $+$20\,pp. CFL+CFO synergy $+$17.5\,pp (super-additive). All-8-ops optimal. \\
\addlinespace
Operator-SENSITIVE & GPT-5.2 & CFL $+$2.5\,pp, CFO $-$5\,pp, CCP 0\,pp. All-8-ops is worst case ($-$10\,pp). CFO must be excluded. \\
\addlinespace
Operator-ROBUST & Sonnet 4 & 6/7 per-operator conditions identical at 80.0\%. Only CFO causes $-$2.5\,pp. Any safe config works. \\
\bottomrule
\end{tabular}}
\end{table}

The archetype taxonomy converts a potential weakness (``no universal best configuration'') into a contribution: empirically characterized, model-specific deployment guidance.
The full per-operator heatmap (Figure~\ref{fig:archetypes-heatmap}) visualizes the contrast across all tested conditions.

\begin{figure}[t]
\centering
\begin{tikzpicture}[
    cell/.style={minimum width=1.45cm, minimum height=0.65cm, anchor=center, font=\scriptsize},
    header/.style={font=\scriptsize\bfseries, anchor=center},
]


\node[header] at (0, 0.5) {};
\node[header, text=tscgblue] at (1.5, 0.5) {Opus 4.7};
\node[header, text=tscgred] at (3.0, 0.5) {GPT-5.2};
\node[header, text=tscggray] at (4.5, 0.5) {Sonnet 4};

\node[header, anchor=east] at (-0.1, 0) {CFL};
\node[cell, fill=tscggreen!25] at (1.5, 0) {$+$5.0};
\node[cell, fill=tscggreen!15] at (3.0, 0) {$+$2.5};
\node[cell, fill=gray!8] at (4.5, 0) {0.0};

\node[header, anchor=east] at (-0.1, -0.7) {CFO};
\node[cell, fill=tscggreen!35] at (1.5, -0.7) {$+$7.5};
\node[cell, fill=tscgred!35] at (3.0, -0.7) {$-$5.0};
\node[cell, fill=tscgred!15] at (4.5, -0.7) {$-$2.5};

\node[header, anchor=east] at (-0.1, -1.4) {CCP};
\node[cell, fill=tscggreen!80] at (1.5, -1.4) {\textbf{$+$20.0}};
\node[cell, fill=gray!8] at (3.0, -1.4) {0.0};
\node[cell, fill=gray!8] at (4.5, -1.4) {0.0};

\node[header, anchor=east] at (-0.1, -2.1) {SAD};
\node[cell, fill=tscggreen!60] at (1.5, -2.1) {$+$15.0};
\node[cell, fill=tscgred!15] at (3.0, -2.1) {$-$2.5};
\node[cell, fill=gray!8] at (4.5, -2.1) {0.0};

\node[header, anchor=east] at (-0.1, -2.8) {CFL+CFO};
\node[cell, fill=tscggreen!70] at (1.5, -2.8) {$+$17.5};
\node[cell, fill=tscgred!15] at (3.0, -2.8) {$-$2.5};
\node[cell, fill=gray!8] at (4.5, -2.8) {0.0};

\node[header, anchor=east] at (-0.1, -3.5) {all-8};
\node[cell, fill=tscggreen!80] at (1.5, -3.5) {\textbf{$+$20.0}};
\node[cell, fill=tscgred!60] at (3.0, -3.5) {\textbf{$-$10.0}};
\node[cell, fill=gray!8] at (4.5, -3.5) {0.0};

\draw[gray!40] (0.7, 0.85) -- (0.7, -3.85);
\draw[gray!40] (2.25, 0.85) -- (2.25, -3.85);
\draw[gray!40] (3.75, 0.85) -- (3.75, -3.85);
\draw[gray!40] (5.25, 0.85) -- (5.25, -3.85);
\foreach \y in {0.3, -0.35, -1.05, -1.75, -2.45, -3.15, -3.85} {
    \draw[gray!40] (0.7, \y) -- (5.25, \y);
}

\node[font=\tiny, text=tscgblue] at (1.5, -4.2) {HUNGRY};
\node[font=\tiny, text=tscgred] at (3.0, -4.2) {SENSITIVE};
\node[font=\tiny, text=tscggray] at (4.5, -4.2) {ROBUST};

\end{tikzpicture}
\caption{Per-operator isolation heatmap ($\Delta$\,pp vs.\ baseline-v10, 43-tool MCP catalog). Opus~4.7 benefits from every operator (HUNGRY); GPT-5.2 is harmed by CFO and all-8-ops (SENSITIVE); Sonnet~4 is invariant except for CFO (ROBUST). Values from leave-one-in methodology ($n{=}40$ per cell, 2~seeds).}
\label{fig:archetypes-heatmap}
\end{figure}

\subsection{Opus Scaling: Heavy-Schema Rescue (T1.7)}
\label{sec:scaling}

Five-point Opus~4.7 scaling curve across synthetic and production schemas (Table~\ref{tab:opus-scaling}).

\begin{table}[!tbp]
\centering
\small
\caption{Opus 4.7 scaling: TSCG advantage persists on heavy production schemas despite saturation on light synthetic catalogs.}
\label{tab:opus-scaling}
\resizebox{\columnwidth}{!}{%
\begin{tabular}{@{}llrrrr@{}}
\toprule
\textbf{Scale} & \textbf{Schema Type} & \textbf{json-text} & \textbf{v13-smart} & \textbf{$\Delta$} & \textbf{Savings} \\
\midrule
43t MCP       & light (benchmark)  & 76.7 & 80.0 & $+$3.3 & 56.6\% \\
50t synthetic & light (synthetic)  & 68.3 & 76.7 & $+$8.3 & 55.1\% \\
75t synthetic & light (synthetic)  & 73.3 & 73.3 &   0.0  & 56.1\% \\
100t synthetic& light (synthetic)  & 75.0 & 75.0 &   0.0  & 52.5\% \\
\textbf{43t MCP}       & \textbf{heavy (production)} & \textbf{75.0} & \textbf{80.0} & \textbf{$+$5.0} & \textbf{56.6\%} \\
\bottomrule
\end{tabular}}
\end{table}

TSCG advantage saturates on light synthetic catalogs at 75--100 tools but \emph{persists} on heavy production MCP schemas ($+$5.0\,pp at {$\sim$}10{,}500 input tokens per call, with perfect seed stability: 80/80/80; Figure~\ref{fig:opus-scaling}).
Token savings remain 52--57\% throughout.
This indicates saturation is catalog-weight-specific, not a fundamental TSCG limit: heavy per-tool schemas carry more parser-ambiguity that structured compilation resolves.

\begin{figure}[t]
\centering
\begin{tikzpicture}
\begin{axis}[
    width=0.85\columnwidth,
    height=5.2cm,
    xlabel={Catalog size (tools)},
    ylabel={TSCG advantage ($\Delta$\,pp)},
    xmin=35, xmax=110,
    ymin=-2, ymax=10,
    xtick={43,50,75,100},
    ytick={-2,0,2,4,6,8,10},
    grid=major,
    grid style={gray!30},
    legend style={
        at={(0.02,0.98)},
        anchor=north west,
        font=\scriptsize,
        draw=gray!50,
        fill=white,
        fill opacity=0.9,
    },
    every axis plot/.append style={thick},
    clip=false,
]

\addplot[
    color=tscgblue,
    mark=o,
    mark size=3pt,
    mark options={solid, fill=tscgblue!30},
] coordinates {
    (43, 3.3)
    (50, 8.3)
    (75, 0.0)
    (100, 0.0)
};
\addlegendentry{Light schemas}

\addplot[
    color=tscgred,
    mark=square*,
    mark size=4pt,
    mark options={solid, fill=tscgred!60},
    only marks,
] coordinates {
    (43, 5.0)
};
\addlegendentry{Heavy MCP (prod.)}

\draw[gray, dashed, thin] (axis cs:35,0) -- (axis cs:110,0);

\draw[tscggray, decorate, decoration={brace, amplitude=4pt, mirror}]
    (axis cs:75, -1.8) -- (axis cs:100, -1.8)
    node[midway, below=5pt, font=\scriptsize, text=tscggray] {saturation};

\draw[-{Stealth[length=4pt]}, tscgred, thick]
    (axis cs:50, 5.0) -- (axis cs:44.5, 5.0)
    node[pos=0, right, font=\scriptsize, text=tscgred] {$+$5.0\,pp};

\end{axis}
\end{tikzpicture}
\caption{Opus~4.7 TSCG advantage across catalog sizes. Light synthetic schemas saturate at 75--100 tools ($\Delta{=}0$), but heavy production MCP schemas (${\sim}$10{,}500 input tokens per call) sustain a $+$5.0\,pp advantage at 43 tools with perfect seed stability (80/80/80). Saturation is catalog-weight-specific, not a fundamental TSCG limit.}
\label{fig:opus-scaling}
\end{figure}
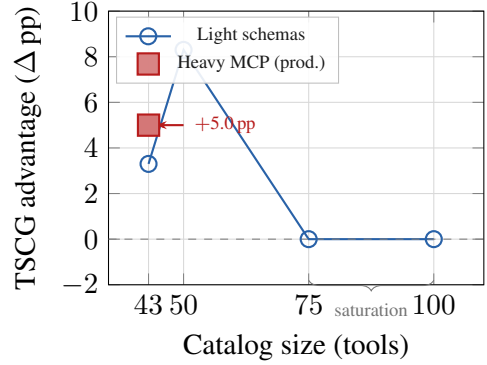

\subsection{Benchmark Generalization: TAB $\to$ MCP (T1.3)}
\label{sec:generalization}

Sonnet~4 on 43-tool MCP Combined: TSCG delta = $-$1.6\,pp on synthetic TAB vs.\ $-$1.7\,pp on real MCP schemas---the synthetic benchmark predicts real MCP behavior within 0.1\,pp.
This directly addresses the ``self-constructed benchmark'' concern with tight quantitative evidence.

\subsection{External Validation}
\label{sec:bfcl-integration}

BFCL: 93.2\% on Berkeley Function Calling Leaderboard schemas (vs.\ 85.7\% natural; ARR~108\%, 46.8\% savings)---TSCG \emph{improves} accuracy on third-party benchmarks.
Naive truncation at 16~tools matches TSCG (87.0\% vs.\ 87.8\%); at 50~ambiguous tools, TSCG achieves 100\% vs.\ 98.5\%, with multi-tool sequencing driving the gap (87.5\% vs.\ 75.0\%).

\FloatBarrier
\section{Discussion}
\label{sec:discussion}


\subsection{Three Mechanisms of Improvement}
\label{sec:three-mechanisms}

TSCG's accuracy improvements arise from three distinct mechanisms.
\textbf{M1: Format Translation} converts JSON schemas to structured text---the dominant mechanism for Class~1 models, where JSON parsing is the bottleneck (e.g., Phi-4: 0\%$\to$84.4\% at 20 tools, entirely format-driven per E4 decomposition).
\textbf{M2: Structural Reorganization} (CAS reordering, CFL constraint positioning, CFO causal ordering) improves accuracy beyond format change---the dominant mechanism for Class~2 frontier models ($+$10.9\,pp mean against native function calling).
However, CFL and CFO become counterproductive at ${\geq}$43 tools, limiting M2's benefit to CAS at larger catalog sizes.
\textbf{M3: Token Reduction} decreases attention dilution as a secondary effect. Notably, GPT-4o in Scenario~A achieves $+$1.0\,pp accuracy improvement with only 6.2\% token savings (Table~\ref{tab:tab-frontier}), suggesting M2 (structural reorganization) operates substantially independently of M3. The relative contribution of M2 versus M3 is a direction for future work with controlled ablations.

\subsection{Four-Class Behavioral Taxonomy}
\label{sec:enabler-optimizer}

Two complementary taxonomies emerge from our analysis. The three-archetype classification in Section~\ref{sec:archetypes} characterizes per-operator response profiles across three frontier models (Opus 4.7, GPT-5.2, Sonnet~4), derived from leave-one-in isolation experiments. The four-class taxonomy below generalizes this analysis to all twelve tested models, combining the archetype data with E4 text-baseline experiments and N3 conservative ablations to capture format-versus-compression behavior across model scales.

E4 text-baseline experiments (2{,}520 calls), N3 conservative ablation (180 calls), and N1 30B benchmarks (840 calls) yield a four-class taxonomy across 12~models:
\textbf{Class~1} (Format-dominated: Phi-4, Mistral 7B, Gemma 4B, Qwen3 4B)---large JSON-baseline gains ($+$17 to $+$90\,pp) that vanish against text baselines ($-$7 to $-$23\,pp); benefit is purely format translation.
\textbf{Class~2} (Compression: Claude, GPT-4o, GPT-5.2)---genuine structural compression persists against text baselines ($+$5 to $+$11\,pp for Claude Sonnet~4).
\textbf{Class~3} (Neutral: Llama 8B, Gemma 12B, Mistral-Small 24B)---TSCG neither helps nor harms.
\textbf{Class~4} (Conservative-only: Qwen3 14B, Qwen2.5-Coder 32B)---balanced TSCG degrades accuracy, but conservative SDM improves it ($+$4.4\,pp).

\paragraph{Deployment guidance.}
Deployment should follow per-model configuration rather than a single universal default.
Opus-class models tolerate the full operator pipeline; GPT-class requires operator selection (excluding CFO); Sonnet-class is robust to any safe configuration.
For local models: conservative profile for all architectures except Qwen (where native JSON is preferable when available); balanced for frontier APIs.
The conservative profile is validated as safe default across 8/9 local models (4B--32B).

\paragraph{Functional validity.}
A spot-check against the Filesystem MCP server confirms that TSCG-compiled schemas produce executable tool calls on production endpoints, validating that compression preserves call-level correctness beyond benchmark accuracy.

\subsection{Why Schema Compression Cannot Be Internalized}
\label{sec:pretokenization}

A natural question is whether tool-schema compression could be achieved through prompt instructions---asking the model to ``internally compress'' schemas before reasoning. This approach fails for several architectural reasons.

\textbf{Tokenization is one-shot and pre-model.} The tokenizer operates on the input string before any model computation. By the time the model processes the input, the JSON schema has already been converted to a token sequence. The model cannot retroactively retokenize or remove tokens already in its context window.

\textbf{Compression instructions add tokens.} An instruction to ``compress tool schemas internally'' is itself additional tokens in the context window. Empirically, attempting this approach increases total token consumption rather than reducing it, while providing no measurable accuracy benefit (we observed this in pilot experiments not included in the main results).

\textbf{Format effects are pre-attention.} As shown by our $R^2$ decomposition (Section~\ref{sec:format-effect}), 88\% of token-cost variance comes from format choice (text vs.\ JSON), with operators contributing additional refinement. The format choice happens at the input boundary---the model has no mechanism to ``reframe'' its inputs mid-stream.

\textbf{Implications.} Our compiler-based approach succeeds specifically because it operates \emph{pre-tokenization}. This positions schema compression as an architectural choice rather than a prompt-engineering technique, with three practical consequences: (1) compression must happen externally, before the API call; (2) caching of compiled schemas is straightforward (deterministic compilation); (3) the compression layer is model-agnostic---the same compiled output works across any model that accepts text input. This architectural argument explains why our approach is preferable to model-side prompting strategies and motivates our external compilation library design.

\subsection{Toward a Community Schema Registry}
\label{ssec:future-registry}

Our compiler-based approach enables a forward-looking research direction: a community-curated registry of pre-compiled tool schemas. The taxonomy could span:

\begin{itemize}
\item \textbf{API services}: compiled schemas for Stripe, Slack, Notion, GitHub, and similar widely-used services.
\item \textbf{MCP servers}: pre-compiled versions of standard Model Context Protocol servers \citep{anthropic_mcp_2024} (filesystem, postgres, GitHub MCP, etc.).
\item \textbf{Database operations}: common SQL, NoSQL, and vector database operations.
\item \textbf{Cloud providers}: AWS, GCP, Azure operation schemas.
\item \textbf{Developer tools}: Git, Docker, Kubernetes operations.
\item \textbf{Productivity tools}: calendar, email, task management schemas.
\end{itemize}

Each entry would include compiled versions for multiple archetype profiles (conservative, balanced, aggressive), tested-against-models metadata, benchmark results, and version pinning relative to the source API.

\paragraph{Quality assurance considerations.}
A schema registry requires CI/CD validation that compiled schemas are semantically equivalent to source schemas. This involves: schema-equivalence testing (the compiled schema must invoke the same functions with the same parameters), accuracy benchmarking against the original (to detect regressions when source APIs evolve), and per-model archetype validation (to ensure profile recommendations remain accurate as models update).

\paragraph{Network effects and maintenance.}
A successful registry would benefit from network effects: as adoption grows, contributors add new schemas, the catalog becomes more comprehensive, and the marginal value of each new contribution increases. However, this creates maintenance obligations: API drift requires re-compilation when source schemas change, and per-version model inversions (Section~\ref{sec:archetypes}) imply profiles need revalidation when major model versions release.

\paragraph{Ecosystem implications.}
Such a registry could establish schema compression as standard infrastructure rather than an opt-in optimization, analogous to package registries in software ecosystems (npm, PyPI). This represents a qualitatively different contribution from a single tool: it would catalyze ecosystem-level efficiency improvements that no single compression library can achieve in isolation. Future work will explore the technical, organizational, and community-governance challenges involved in building and maintaining such infrastructure.

\section{Limitations}
\label{sec:limitations}

\textbf{Benchmark scope.} TAB is self-constructed; we mitigate this with BFCL external validation (ARR 108\%) and TAB-to-MCP transfer (0.1\,pp delta), but independent third-party evaluation on diverse catalogs is needed.

\textbf{Task coverage.} TAB measures tool selection accuracy and parameter extraction only---not generation quality, multi-turn coherence, or end-to-end task completion.

\textbf{Small-model degradation.} Class~1 models show negative compression gains against text baselines ($-$7 to $-$23\,pp), meaning TSCG harms accuracy in hypothetical text-mode deployments with $<$10 tools.

\textbf{Language.} All evaluations and filler patterns are English-only; multilingual extension requires language-specific libraries.

\textbf{Model coverage.} 12 models across 4 architecture families. Generalization to untested architectures (Mamba, RWKV) and thinking models (o1, o3-mini) is uncertain.

\textbf{Model version evolution.} Frontier model APIs are continuously updated; exact performance characterization holds for the tested API versions during March--April 2026. Results may shift as providers deploy model updates.

\textbf{Predictive model.} LOO RMSE of 12.8\,pp provides coarse triage; Phi-4's zero-accuracy leverage points inflate $R^2$ (excluding Phi-4: $R^2{=}0.81$).

\textbf{Scope and choice of approach.} Our compiler-based, format-preserving approach achieves moderate compression (50--72\%) compared to schema-replacement strategies like MCP Code Mode \citep{anthropic_mcp_code_2025, cloudflare_code_mode_2026}, which achieve 98--99\% reduction. The choice depends on use case: schema replacement excels when all tools have executable code interfaces and a sandboxed runtime is acceptable. Our approach preserves full schema semantics and applies to non-executable tool definitions; we observe this fidelity to be beneficial when models can use parameter constraints during reasoning. Compared to mcp-compressor's wrapper-tool reduction \citep{atlassian_mcp_compressor_2026} (70--97\%), our compression is lower but maintains full schema visibility, which we hypothesize underlies our consistent accuracy improvements (Section~\ref{sec:results}). A direct head-to-head benchmark against these systems is left to future work.

\textbf{Evaluation freeze.} GPT-5.5 was released after our evaluation period (March--April 2026). Extending TAB to post-freeze models is left to future work.

\textbf{AI assistance.} AI language models assisted with code quality checks, statistical verification, and manuscript editing. All experimental design, data collection, analysis, and scientific conclusions are the sole work of the authors.

\FloatBarrier
\section{Conclusion}
\label{sec:conclusion}

We have presented Token-Context Semantic Grammar (TSCG), a deterministic tool-schema compiler that addresses the protocol mismatch between JSON-based tool APIs and LLM token processing.

\textbf{JSON-API enablement.} Every production framework transmits tools as JSON, yet small models (4B--14B) achieve only 0--49\% accuracy at $>$15 JSON-format tools.
TSCG recovers accuracy to 65--90\% by translating JSON schemas to structured text---enabling small models as functional tool-use agents in edge and privacy-constrained environments.
E4 decomposition (2{,}520 Ollama calls across 6 models) identifies format translation as the dominant mechanism for this class.

\textbf{Frontier compression.} For frontier models (Claude Sonnet~4, GPT-4o, GPT-5.2), TSCG provides genuine structural compression ($+$5--11\,pp for Claude Sonnet~4), with 50--72\% token savings that persist when the format confound is eliminated.
On heavy production MCP schemas, advantages persist ($+$5.0\,pp) despite saturation on light synthetic catalogs.

\textbf{Mechanistic decomposition.} The first format-versus-compression analysis for tool-schema optimization yields a four-class taxonomy across 12 models (4B--32B + frontier), with concrete deployment recommendations per model class.
The predictive model ($R^2{=}0.88$, collapsing to 0.03 against text baselines) quantifies format sensitivity and enables deployment triage from a single baseline measurement.

TSCG is released as an open-source npm package---1{,}200 LOC TypeScript, zero dependencies, sub-millisecond execution.\footnote{\url{https://github.com/SKZL-AI/tscg}}
TAB, the first tool-schema compression benchmark (\tscgTotalCalls{} calls, 12 models, 5 scenarios), is released as open-source evaluation infrastructure.

\section*{Acknowledgements}
This work was conducted independently without institutional funding or affiliation.
The author thanks Dominic Wolf\footnote{ORCID: \url{https://orcid.org/0009-0003-5789-1991}} for sharing valuable practical experience with TSCG deployments that informed several design decisions.
Portions of the experimental analysis, quality control, and manuscript preparation were assisted by AI language models (Claude, GPT-4), consistent with ACL's policy on AI writing assistance.
All scientific claims, experimental design, data analysis, and final editorial decisions are solely the author's responsibility.

\section*{Data and Code Availability}
The TSCG compiler, TAB benchmark suite (task definitions, tool schemas, evaluation harness), and all evaluation scripts are released as open-source software at: \url{https://github.com/SKZL-AI/tscg}.
A persistent archive is available via Zenodo: \url{https://doi.org/10.5281/zenodo.19795991}.
Raw evaluation logs are available on request.

\bibliography{references}

\FloatBarrier
\appendix

\section{Pre-TAB Evaluation}
\label{app:pre-tab}

The pre-TAB evaluation comprised 19 core tasks across 7 categories (Factual $n{=}4$, Reasoning $n{=}4$, Classification $n{=}2$, Extraction $n{=}1$, OptFirst $n{=}3$, Complex $n{=}2$, NearDup $n{=}3$), 6 compression strategies, 11 independent runs on Claude Sonnet~4, and domain-specific benchmarks across 4 phases. This evaluation established TSCG's baseline effectiveness before the multi-model TAB benchmark. Full results are available in the supplementary materials.

\section{GSM8K-Under-Load}
\label{app:gsm8k}

Tool-schema overhead does not degrade reasoning: Claude Sonnet~4 maintains ${\sim}$81\% GSM8K accuracy across 0--50 irrelevant tool definitions (${\sim}$400 calls across both models). TSCG provides a consistent ${\sim}$4.5\,pp advantage over natural schemas for Claude Sonnet~4, suggesting that reduced schema overhead frees attention capacity for mathematical reasoning. The effect is statistically significant ($+$4.5\,pp, $p = 0.020$) but practically small (Cohen's $d = 0.11$). GPT-5.2 shows slightly more degradation under schema load but maintains $>$75\% accuracy throughout.

\section{SDM Ablation Detail}
\label{app:sdm-detail}

Conservative SDM (filler removal only) is compared with Balanced (full structural compression) on Qwen3-14B (180~calls, N3 corrected re-run at sizes 10/20/50):

\begin{itemize}[noitemsep]
    \item Balanced TSCG degrades accuracy by $-$6.6\,pp mean (90.2\%$\to$84.1\% at 20 tools, 94.6\%$\to$89.6\% at 50 tools).
    \item Conservative SDM \emph{improves} accuracy by $+$4.4\,pp mean (90.2\%$\to$99.3\% at 20 tools, 94.6\%$\to$95.0\% at 50 tools).
    \item CAS reordering and bracket elision---not SDM filler removal---trigger the degradation.
\end{itemize}

\noindent Practical recommendation: Conservative SDM for all models ${<}$10B and all Qwen architectures (including Qwen2.5-Coder~32B, confirmed by N1); Balanced for frontier models where genuine compression persists.

\section{Predictive Model and LOO Cross-Validation}
\label{app:loo-details}

The accuracy gap $\Delta(m)$ between Natural and TSCG for model $m$ can be predicted from Natural-baseline performance alone:
\begin{equation}
\label{eq:correlation}
\Delta(m) = \alpha \cdot \text{Natural}(m) + \beta + \epsilon
\end{equation}
Table~\ref{tab:r2-robustness} reports robustness checks on this regression.

\begin{table}[t]
\centering
\small
\caption{R\textsuperscript{2} robustness checks for the accuracy-gap regression.}
\label{tab:r2-robustness}
\begin{tabular}{@{}lcccc@{}}
\toprule
\textbf{Configuration} & $n$ & $R^2$ & \textbf{95\% CI} & $p$ \\
\midrule
7-model full set & 49 & 0.885 & $[0.78,\; 0.93]$ & $< 10^{-20}$ \\
7-model sans Phi-4 & 42 & 0.813 & $[0.58,\; 0.94]$ & $< 10^{-10}$ \\
Model-level (7 means) & 7 & 0.951 & $[0.62,\; 1.00]$ & $< 0.001$ \\
Text baseline (E4) & 49 & 0.028 & --- & 0.24 \\
\bottomrule
\end{tabular}
\end{table}

\noindent Phi-4~14B is the strongest leverage point ($\Delta R^2 = -0.07$ when removed), but $R^2$ remains above 0.80 without it. The linear relationship holds at the model level ($R^2 = 0.95$). LOO cross-validation (holding out one model, refitting on remaining six) yields RMSE~=~12.8\,pp---adequate for deployment triage but insufficient for precise predictions. Qwen3~4B is the only model where the prediction direction is wrong (predicted positive, actual negative), reflecting its volatile behavior.

\section{Theorem Proofs}
\label{app:proofs}

\begin{proof}[Proof of Theorem~\ref{thm:compression}]
Each token-reducing operator $T_i \in \mathcal{T}_R$ acts on a disjoint subset of tokens:
\begin{itemize}[noitemsep]
    \item \textbf{SDM}: filler patterns (104+ curated, non-overlapping with delimiters/tokenizer targets).
    \item \textbf{DRO}: verbose delimiter phrases (structural role markers).
    \item \textbf{TAS}: suboptimal tokenizer splits (multi-token $\to$ single-token Unicode).
    \item \textbf{CFL}: positional relocation of output constraint to position~0.
\end{itemize}
Disjoint action $\implies$ additive composition. $\mathcal{T}_S$ preserves count; $B{=}0$ excludes expansion:
\[
    |\tok(\Pi(S))| \;\leq\; |\tok(S)| \cdot \Bigl(1 - \sum_{T_i \in \mathcal{T}_R} r_i \cdot f_i(S)\Bigr). \qedhere
\]
\end{proof}

\section{Implementation and Deployment}
\label{app:deployment}

TSCG is implemented as ${\sim}$1{,}200 lines of TypeScript with zero external dependencies. Three deployment modes: CLI (npm), Chrome Extension (Manifest~V3), and Web Application (React + Vite). Production bundle: 34.7\,KB (11.7\,KB gzipped). Sub-millisecond execution for prompts under 4{,}096 tokens. Four optimization profiles (Table~\ref{tab:profiles}):

\begin{table}[t]
\centering
\small
\caption{TSCG optimization profiles as implemented in \texttt{@tscg/core} v1.4.0.}
\label{tab:profiles}
\begin{tabular}{@{}lp{6cm}@{}}
\toprule
\textbf{Profile} & \textbf{Active Principles} \\
\midrule
\texttt{conservative} & SDM only. Safe default for all local models (4B--32B). \\
\texttt{balanced} & SDM, CAS, CFO, DRO, TAS, CCP. CFL/CFO auto-disabled at ${\geq}$30 tools in the released implementation. \\
\texttt{aggressive} & All 8 on Claude; 6 on non-Claude (CFL and SAD-F auto-disabled by echo-back guard). \\
\texttt{auto} & Adaptive: conservative at ${\leq}$20 tools, balanced (sans CFL/CFO) at 21--40, conservative at $>$40. \\
\bottomrule
\end{tabular}
\end{table}

\section{Extended Prior-Art Comparison}
\label{app:prior-art}

\begin{table*}[t]
\centering
\small
\setlength{\tabcolsep}{2.5pt}
\begin{tabular}{@{}l c c c c c c c c@{}}
\toprule
\textbf{System} & \textbf{Det.} & \textbf{Tok.} & \textbf{Caus.} & \textbf{B-Box} & \textbf{Bdgt.} & \textbf{Cmp.} & \textbf{NoTr.} & \textbf{0-Dep} \\
\midrule
\multicolumn{9}{l}{\textit{Compression-Based}} \\
\cmidrule(r){1-1}
LLMLingua \cite{jiang2023llmlingua}         & \xmark & \xmark & \xmark & \cmark & \cmark & \cmark & \xmark & \xmark \\
LLMLingua-2 \cite{pan2024llmlingua2}        & \xmark & \xmark & \xmark & \cmark & \cmark & \cmark & \xmark & \xmark \\
Selective Context \cite{li2023compressing}   & \cmark & \xmark & \xmark & \cmark & \xmark & \cmark & \cmark & \xmark \\
LongLLMLingua \cite{jiang2024longllmlingua}  & \xmark & \xmark & \xmark & \cmark & \cmark & \cmark & \xmark & \xmark \\
Gist Tokens \cite{mu2023learning}            & \xmark & \xmark & \xmark & \xmark & \xmark & \cmark & \xmark & \xmark \\
AutoCompressors \cite{chevalier2023adapting} & \xmark & \xmark & \xmark & \xmark & \xmark & \cmark & \xmark & \xmark \\
LeanContext \cite{arefeen2024leancontext}    & \xmark & \xmark & \xmark & \cmark & \xmark & \cmark & \xmark & \xmark \\
\midrule
\multicolumn{9}{l}{\textit{Search-Based}} \\
\cmidrule(r){1-1}
DSPy \cite{khattab2023dspy}                  & \xmark & \xmark & \xmark & \cmark & \xmark & \xmark & \cmark & \xmark \\
OPRO \cite{yang2024large}                    & \xmark & \xmark & \xmark & \cmark & \xmark & \xmark & \cmark & \xmark \\
APE \cite{zhou2023large}                     & \xmark & \xmark & \xmark & \cmark & \xmark & \xmark & \cmark & \xmark \\
EvoPrompt \cite{chen2024evoprompt}           & \xmark & \xmark & \xmark & \cmark & \xmark & \xmark & \cmark & \xmark \\
PromptBreeder \cite{fernando2024promptbreeder} & \xmark & \xmark & \xmark & \cmark & \xmark & \xmark & \cmark & \xmark \\
USP \cite{wan2024cosp}                       & \xmark & \xmark & \xmark & \cmark & \xmark & \xmark & \cmark & \xmark \\
\midrule
\multicolumn{9}{l}{\textit{Gradient-Based}} \\
\cmidrule(r){1-1}
TextGrad \cite{yuksekgonul2024textgrad}      & \xmark & \xmark & \xmark & \xmark & \xmark & \xmark & \cmark & \xmark \\
ProTeGi \cite{pryzant2023automatically}      & \xmark & \xmark & \xmark & \xmark & \xmark & \xmark & \cmark & \xmark \\
PromptAgent \cite{wang2023promptagent}       & \xmark & \xmark & \xmark & \cmark & \xmark & \xmark & \cmark & \xmark \\
\midrule
\multicolumn{9}{l}{\textit{Template / Theoretical}} \\
\cmidrule(r){1-1}
LangGPT \cite{wang2024langgpt}              & \cmark & \xmark & \xmark & \cmark & \xmark & \xmark & \cmark & \xmark \\
\midrule
\textbf{TSCG (ours)}                         & \cmark & \cmark & \cmark & \cmark & \cmark & \cmark & \cmark & \cmark \\
\bottomrule
\end{tabular}
\caption{Comparison of 16 prompt optimization systems and TSCG across eight desiderata. \textbf{Determ.}: deterministic output for the same input. \textbf{Tok-Aware}: optimizes with respect to BPE tokenizer boundaries. \textbf{Causal Th.}: grounded in causal attention theory. \textbf{Black-Box}: requires no access to model weights or gradients. \textbf{Budg. Anch.}: budget-constrained anchor duplication. \textbf{Compress.}: achieves measurable token compression. \textbf{No Train}: requires no model training or fine-tuning. \textbf{Zero-Dep}: zero external runtime dependencies. TSCG is the only system satisfying all eight criteria.}
\label{tab:prior-art}
\end{table*}

\section{Benchmark Configurations and Statistical Detail}
\label{app:configs}

Table~\ref{tab:run-configs} documents configuration parameters for all benchmark runs.

\begin{table}[t]
\centering
\small
\setlength{\tabcolsep}{2.5pt}
\caption{Experimental configurations.}
\label{tab:run-configs}
\begin{tabular}{@{}llll@{}}
\toprule
\textbf{Parameter} & \textbf{Core/Domain} & \textbf{TAB} & \textbf{Ollama} \\
\midrule
Temperature & 0 & 0 & 0 \\
Seed & N/A & 42 & 42 \\
Max tokens & 4096 & 1024 & 1024 \\
Mode & Text + FC & Text-mode & Text-mode \\
SDM profile & full & balanced & balanced \\
API provider & Anthropic/OpenAI & Anthropic & Ollama (local) \\
Eval. method & Automated & TABEvaluator & TABEvaluator \\
\bottomrule
\end{tabular}
\end{table}

\subsection{Holm-Bonferroni Correction}
\label{app:holm-bonferroni}

107 pairwise McNemar tests, family-wise $\alpha = 0.05$. 9/107 achieve significance (Table~\ref{tab:holm-bonferroni}), all from Scenario~D with weak JSON baselines.

\begin{table}[t]
\centering
\small
\setlength{\tabcolsep}{3pt}
\caption{Significant comparisons (Holm-Bonferroni, $\alpha{=}0.05$, 107 tests).}
\label{tab:holm-bonferroni}
\begin{tabular}{@{}llcrr@{}}
\toprule
\textbf{Model} & \textbf{Tools} & \textbf{$\Delta$pp} & \textbf{$p_{\mathrm{adj}}$} & \textbf{$d$} \\
\midrule
Phi-4 14B    & 5/10/15 & $+$95.0 & 0.003 & 6.11 \\
Phi-4 14B    & 20 & $+$84.4 & 0.003 & 4.21 \\
Phi-4 14B    & 30 & $+$80.0 & 0.011 & 2.81 \\
Phi-4 14B    & 50 & $+$90.3 & 0.003 & 6.11 \\
Gemma 3 4B   & 3  & $+$85.0 & 0.007 & 3.34 \\
Gemma 3 4B   & 50 & $+$63.1 & 0.004 & 1.68 \\
Gemma 3 4B   & 50\textsuperscript{c} & $+$65.0 & 0.004 & 1.73 \\
\bottomrule
\end{tabular}
\end{table}

\subsection{Bootstrap Confidence Intervals}
\label{app:bootstrap-ci}

\begin{table}[t]
\centering
\small
\setlength{\tabcolsep}{3pt}
\caption{Bootstrap 95\% CIs for key deltas (1{,}000 resamples, seed=42).}
\label{tab:bootstrap-ci}
\begin{tabular}{@{}llrrl@{}}
\toprule
\textbf{Scenario} & \textbf{Model} & \textbf{$\Delta$pp} & \textbf{95\% CI} & \textbf{$p$} \\
\midrule
Sc.\,A (16) & Claude S.\,4 & $+$10.0 & [0.6, 18.3] & 0.200 \\
Sc.\,A (16) & GPT-5.2 & $+$29.7 & [16.1, 42.2] & 0.004 \\
Sc.\,B (43) & Claude S.\,4 & $+$5.3 & [$-$0.3, 11.3] & 0.063 \\
Sc.\,D (50) & Phi-4 14B & $+$90.3 & --- & $<$0.001 \\
Sc.\,D (50) & Gemma 3 4B & $+$63.1 & [45.0, 80.0] & $<$0.001 \\
GSM8K (50) & Claude S.\,4 & $+$4.5 & [0.5, 8.5] & 0.020 \\
\bottomrule
\end{tabular}
\end{table}

\section{11-Model Ablation}
\label{app:ablation}
\label{sec:extended-findings} 

\begin{table*}[t]
\centering
\small
\begin{tabular}{@{}l rr rr l@{}}
\toprule
\textbf{Disabled} & \multicolumn{2}{c}{\textbf{$\Delta$ Acc.\,(pp)}} & \multicolumn{2}{c}{\textbf{$\Delta$ Tokens}} & \textbf{Role} \\
\cmidrule(lr){2-3} \cmidrule(lr){4-5}
\textbf{Operator} & \textbf{16\,t} & \textbf{43\,t} & \textbf{16\,t} & \textbf{43\,t} & \\
\midrule
$-$CAS   & $-$1.4 & $-$2.3 & $-$2    & 0      & Accuracy driver \\
$-$SDM   & $-$0.5 & $-$1.8 & $+$1    & $-$28  & Schema-dependent \\
$-$TAS   & $-$0.9 & $+$0.5 & $+$1    & $-$4   & Mixed (scale-sensitive) \\
$-$SAD-F & $+$0.5 & $-$0.5 & $-$9    & $-$20  & Negligible \\
$-$CCP   & $\pm$0.0  & $\pm$0.0  & $-$85   & $-$306 & No effect (token overhead) \\
$-$CFO   & $-$0.5 & $+$1.4 & 0       & 0      & Hurts at ${\geq}$43\,t \\
$-$CFL   & $\pm$0.0  & $+$0.9 & $-$6    & $-$6   & Hurts at ${\geq}$43\,t \\
$-$DRO   & $+$0.5 & $+$0.5 & $+$161  & $+$611 & Token engine (95\% of savings) \\
\bottomrule
\end{tabular}
\caption{Leave-one-out ablation averaged across 11 models (9 local 4B--32B + GPT-4o + GPT-5.2), 20 tasks per condition. Baselines: 65.0\% at 16 tools, 61.8\% at 43 tools. $\Delta$~Acc.: accuracy change when operator is removed (negative = operator helps). $\Delta$~Tokens: absolute token change. CAS is the only operator with consistent cross-model accuracy benefit. DRO drives 95\% of token savings but has no accuracy effect. CFL and CFO become counterproductive at ${\geq}$43 tools. Claude Opus~4.7 (12th model, not averaged) confirms the ranking; its unique CCP sensitivity ($-$5\,pp) is discussed in \S\ref{sec:extended-findings}.}
\label{tab:ablation}
\end{table*}

\paragraph{Opus 4.7.}
Claude Opus~4.7 emerges as the only operator-HUNGRY archetype in our study: every operator contributes positively, with CCP alone delivering $+$20\,pp improvement and CFL+CFO synergizing super-additively ($+$17.5\,pp vs.\ expected $+$12.5\,pp). The full 8-operator pipeline is optimal for Opus-class models.

\end{document}